\documentclass[runningheads,orivec,abstracton]{scrartcl}
\usepackage[utf8]{inputenc}
\usepackage[small]{caption}
\usepackage{authblk}
\usepackage{cite}
\usepackage{subfig}
\usepackage{wrapfig}
\usepackage{xcolor}
\usepackage{xspace}
\usepackage{enumerate}
\usepackage{microtype}

\usepackage{amsmath}
\usepackage{amssymb}
\usepackage{amstext}
\usepackage{esvect}

\usepackage{amsthm}

\newtheorem{observation}{Observation}
\newtheorem{lemma}{Lemma}
\newtheorem{theorem}{Theorem}
\newtheorem{corollary}{Corollary}
\newtheorem{definition}{Definition}
\newtheorem{claim}{Claim}
\newtheorem{property}{Property}
\newtheorem{problem}{Problem}

\usepackage{amsopn}
\usepackage{hyperref}
\usepackage{graphicx}
\usepackage[longend,vlined,linesnumbered,procnumbered]{algorithm2e}
\usepackage{paralist}
\defaultleftmargin{0.0em}{0.0em}{0.0em}{0.0em}
\usepackage[export]{adjustbox}
\usepackage[basic]{complexity}
\usepackage[left, pagewise, mathlines]{lineno}
\newcommand*\patchAmsMathEnvironmentForLineno[1]{%
  \expandafter\let\csname old#1\expandafter\endcsname\csname #1\endcsname
  \expandafter\let\csname oldend#1\expandafter\endcsname\csname end#1\endcsname
  \renewenvironment{#1}%
     {\linenomath\csname old#1\endcsname}%
     {\csname oldend#1\endcsname\endlinenomath}}%

\usepackage[textsize=footnotesize,disable]{todonotes}

\newcommand{\hp}[2]{\ensuremath{l_{#1 #2}^+}}
\newcommand{\eps}{\ensuremath{\varepsilon}}

\newcommand{\sa}{self-ap\-proaching\xspace}

\newcommand{\ic}{increasing-chord\xspace}

\newcommand{\N}{\ensuremath{\mathbb N_0}}

\newcommand{\ray}[2]{\ensuremath{\textnormal{ray}(#1,#2)}}
\newcommand{\dg}{^\circ}
\newcommand{\Wlg}{Without loss of generality}
\newcommand{\wlg}{without loss of generality}

\newcommand{\sublem}[1]{\lowercase\expandafter{(\romannumeral #1)\relax}\xspace}

\renewcommand{\vec}[1]{\vv{#1}}

\newcommand{\parent}[1]{\pi_{#1}}
\newcommand{\rp}[1]{\rho_l^i}
\newcommand{\lp}[1]{\rho_r^i}
\newcommand{\gtd}{GTD\xspace}
\newcommand{\gtds}{GTDs\xspace}

\newcommand{\minmulticut}{\textsc{Minimum Multicut}\xspace}

\newcommand{\nray}[2]{\textnormal{ray}_{#1}(#2)}
\newcommand{\normal}[2]{\textnormal{n}_{#1}(#2)}
\newcommand{\nproj}[2]{\textnormal{proj}_{#1}(#2)}

\newenvironment{reptheorem}[2]{%
        \par\addvspace{8pt plus3pt minus2pt}%
   \noindent{\textbf{#1 #2.}}%
   \hskip.5em\ignorespaces
   }{%
        \par\addvspace{8pt plus3pt minus2pt}}

\DeclareMathOperator{\opt}{OPT}

\makeatletter
\providecommand*{\cupdot}{%
  \mathbin{%
    \mathpalette\@cupdot{}%
  }%
}
\newcommand*{\@cupdot}[2]{%
  \ooalign{%
    $\m@th#1\cup$\cr
    \hidewidth$\m@th#1\cdot$\hidewidth
  }%
}
\makeatother

\begin{document}

\title{Partitioning Graph Drawings\\ and Triangulated Simple Polygons\\ into Greedily Routable Regions\thanks{A preliminary version of this paper has been presented at the 26th
  International Symposium on Algorithms and Computation (ISAAC
  2015)~\cite{Noellenburg2015}.}}
\author[1]{Martin Nöllenburg}
\author[2]{Roman Prutkin}
\author[3]{Ignaz Rutter}

\affil[1]{Algorithms and Complexity Group, TU
  Wien, Vienna, Austria}
\affil[2]{Institute of Theoretical
  Informatics, Karlsruhe Institute of Technology, Germany}
\affil[3]{Algorithms and Visualization W\&I, Technische Universiteit
  Eindhoven, Eindhoven, Netherlands}

\date{}
\maketitle

\begin{abstract}
  A greedily routable region (GRR) is a closed subset
  of~$\mathbb R^2$, in which any destination point can be reached from
  any starting point by always moving in the direction with maximum
  reduction of the distance to the destination in each point of the
  path.
  Recently, Tan and Kermarrec proposed a geographic routing protocol
  for dense wireless sensor networks based on decomposing the network
  area into a small number of interior-disjoint GRRs. They showed that
  minimum decomposition is \NP-hard for polygonal regions with holes.

  We consider minimum GRR decomposition for plane straight-line
  drawings of graphs. Here, GRRs coincide with self-approaching
  drawings of trees, a drawing style which has become a popular
  research topic in graph drawing.  We show that minimum decomposition
  is still \NP-hard for graphs with cycles and even for trees, but can
  be solved optimally for trees in polynomial time, if we allow only
  certain types of GRR contacts. Additionally, we give a
  2-approximation for simple polygons, if a given triangulation has to
  be respected.

\end{abstract}

\section{Introduction}
\label{sec:introduction}

Geographic or geometric routing is a routing approach for wireless
sensor networks that became popular recently. It uses geographic
coordinates of sensor nodes to route messages between them. 
One simple routing strategy is greedy routing. Upon receipt of a
message, a node tries to forward it to a neighbor node that is closer
to the destination than itself. However, delivery cannot be
guaranteed, since a message may get stuck in a local minimum or \emph{void}.
Another local routing strategy is compass routing. It forwards the
message to a neighbor, such that the direction from the node to this
neighbor is closest to the direction from the node to the
destination. Kranakis et al.~\cite{kranakis_compass_1999} showed that
compass routing can produce loops even in plane triangulations. They
also showed that compass routing is always successful on Delaunay
triangulations.
More advanced geometric routing protocols employ strategies like face
routing~\cite{bmsu-rgdahwn-01} and related techniques based on planar
graphs to get out of local minima; see~\cite{mwh-spbr-01, cv-svht-07}
for an overview.

An alternative approach is to decompose the network into components
such that in each of them greedy routing is likely to perform
well~\cite{Glider2005, tbk-cpsn-09, zsg-ssasn-07}. A global data
structure of preferably small size is used to store interconnectivity
between components.
One such network decomposition approach has been recently proposed by
Tan and Kermarrec~\cite{tk-ggrlssn-2012}.
They assume that \emph{global} connectivity irregularities, i.e, large
holes in the network and the network boundary, are the main source of
local minima in which greedy routing between a pair of sensor nodes
might get stuck. They note that in practical sensor networks,
\emph{local} connectivity irregularity normally has low impact on the
cost of routing and the quality of the resulting paths, since the
local minima in this context can be overcome by simple and
light-weight techniques; see~\cite{tk-ggrlssn-2012} for a list of such
strategies. With this reasoning, Tan and Kermarrec model the network
as a polygonal region with obstacles or holes inside it and consider
greedy routing inside this continuous domain. Local minima now only
appear on the boundaries of the polygonal region. In this work, we use
the same model.

Tan and Kermarrec~\cite{tk-ggrlssn-2012} try to partition this region
into a minimum number of polygons, in which greedy routing works
between any pair of points. They call such components \emph{greedily
  routable regions (GRRs)}. For intercomponent routing, region
adjacencies are stored in a graph. The protocol is able to guarantee
finding paths of \emph{bounded stretch}, i.e., the length of such a
path exceeds the distance between its endpoints only by a constant
factor.

For routing in the underlying network of sensor nodes corresponding to
discrete points inside the polygonal region, greedy routing is used if
the source and the destination nodes are in the same component, and
existing techniques are used to overcome local minima. For
inter-component routing, each node stores a neighbor on a shortest path
to each component. This path is used to get to the component of the
destination, and then intra-component routing is used.

Tan and Kermarrec~\cite{tk-ggrlssn-2012} emphasize the importance for
the nodes to store as small routing tables as possible and note that
the number of network components in a decomposition directly reflects
the number of nonlocal routing states of a node. This number
determines the size of the node's routing table. Therefore, the goal
is to partition the network into a minimum number of GRRs. In this
work, we focus on the problem of partitioning a polygonal region or a
graph drawing (for which we extend the notion of a GRR) into a minimum
number of GRRs. For a detailed description of an actual routing
protocol based on GRR decompositions, see the original work of Tan and
Kermarrec~\cite{tk-ggrlssn-2012}.

The authors prove that partitioning a polygon with holes into a
minimum number of regions is \NP-hard and they propose a simple
heuristic. Its solution may strongly deviate from the optimum even for
very simple polygons; see Fig.~\ref{fig:heuristic-bad}.

Some real-world instances from the work of~Tan and
Kermarrec~\cite[Fig.~17]{tk-ggrlssn-2012} are networks of sensor nodes
distributed on roads of a city. The resulting polygonal regions are
very narrow and strongly resemble plane straight-line graph
drawings. Therefore, considering plane straight-line graph drawings in
addition to polygonal regions is a natural adjustment of the minimum
GRR partition problem.

In this paper, we approach the problem of finding minimum or
approximately minimum GRR decompositions by first considering the
special case of partitioning drawings of graphs, which can be
interpreted as very thin polygonal regions. We notice that in this
scenario, GRRs coincide with \ic drawings of trees as studied by
Alamdari et al.~\cite{acglp-sag-12}.

A \emph{self-approaching} curve is a curve, where for any point~$t'$
on the curve, the Euclidean distance to~$t'$ decreases continuously
while traversing the curve from the start to~$t'$~\cite{ikl-sac-99}.
An \emph{increasing-chord} curve is a curve that is self-approaching
in both directions.  The name is motivated by their equivalent
characterization as those curves, where for any four points $a,b,c,d$
in this order along the curve, $|bc| \le |ad|$, where~$|pq|$ denotes
the Euclidean distance from point~$p$ to point~$q$.

A graph drawing is \emph{\sa} or \emph{\ic} if every pair of vertices
is joined by a \sa or \ic path, respectively.  The study of \sa and
\ic graph drawings was initiated by Alamdari et
al.~\cite{acglp-sag-12}.  They studied the problem of recognizing
whether a given graph drawing is self-approaching and gave a complete
characterization of trees admitting self-approaching drawings.  In our
own previous work~\cite{npr-sid3pg-16}, we studied \sa and \ic
drawings of triangulations and 3-connected planar graphs.
Furthermore, the problem of connecting given points to form an \ic
drawing has been investigated~\cite{acglp-sag-12,dfg-icgps-15}.

\paragraph{Contributions.} First, we show that partitioning a plane
graph drawing into a minimum number of \ic components is
\NP-hard. This extends the result of Tan and
Kermarrec~\cite{tk-ggrlssn-2012} for polygonal regions with holes to
plane straight-line graph drawings. Next, we consider plane drawings
of trees. We show that the problem remains \NP-hard even for trees, if
arbitrary types of GRR contacts are allowed. For a restriction on the
types of GRR contacts, we show how to model the decomposition problem
using \minmulticut, which provides a polynomial-time
2-approximation. We then solve the partitioning problem for trees and
restricted GRR contacts optimally in polynomial time
using dynamic programming.
Finally, we use the insights gained for decomposing graphs and apply
them to the problem of minimally decomposing simple triangulated
polygons into GRRs.  We provide a polynomial-time
2-approximation for decompositions that are formed along chords of the
triangulation.

\section{Preliminaries}

\newcommand{\vis}[1]{V(#1)} \newcommand{\forwtan}[1]{\vec{d_{#1}}}
\newcommand{\polyg}{\mathcal P}
\newcommand{\vangle}[2]{\angle(#1,#2)}
In the following, let~$\mathcal P$ be a polygonal region, and
let~$\partial \mathcal P$ denote its boundary.  For
$p \in \mathcal P$, let~$\vis{p}$ denote
the~$\emph{visibility region}$ of~$p$, i.e., the set of
points~$q \in \mathcal P$ such that the line segment $p q$ lies
inside~$\mathcal P$.
For directions~$\vec {d_1}$ and~$\vec {d_2}$,
let~$\vangle{\vec {d_1}}{\vec {d_2}} \leq 180\dg$ denote the angle
between them. For points~$p,q$, $p \neq q$, let~$\ray{p}{q}$ denote
the ray with origin~$p$ and direction~$\vec{p q}$.

\begin{definition}
  For an~$s$-$t$-path~$\rho$ and a point~$p \neq t$ on~$\rho$, we define
  the \emph{forward tangent} on~$\rho$ in~$p$ as the direction
  $\vec{d} = \lim_{\eps \rightarrow 0} \{ \vec{pq} \mid q \textnormal{ succeeds
  } p \textnormal{ on } \rho \textnormal {, and } |pq|=\eps \}$.
\end{definition}

Next, we formally define paths resulting from greedy routing
inside~$\mathcal P$. We call such paths~\emph{greedy}. Note that this
definition of greediness is different from the one used in the context
of greedy embeddings of graphs~\cite{Papadimitriou2005}.

\begin{definition}
  For points~$s,t \in \mathcal P$, an~$s$-$t$-path~$\rho$ is
  \emph{greedy} if the distance to~$t$ strictly decreases along~$\rho$
  and if for every point~$s' \neq t$ on~$\rho$, the forward
  tangent~$\vec{d}$ on~$\rho$ in~$s'$ has the
  minimum angle with~$\vec{s't}$ among all vectors~$\vec{s'q}$ for
  any~$q \in \vis{s'} \setminus \{s' \}$.
  \label{def:greedy-path}
\end{definition}

\begin{figure}[tb] \hfill
  \subfloat[]{\includegraphics[scale=1.0,page=1]{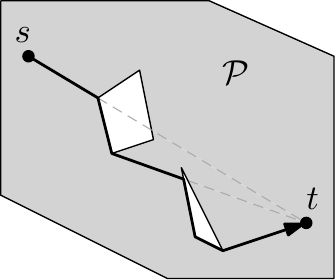}\label{fig:greedy-path}}\hfill
  \subfloat[]{\includegraphics[scale=1.0,page=2]{fig/greedy-path.pdf}\label{fig:greedy-path-alg}}
 \hfill\null
 \caption{\protect\subref{fig:greedy-path}~The thick $s$-$t$-path
   inside the polygonal region~$\mathcal P$ (grey) is
   greedy. \protect\subref{fig:greedy-path-alg}~If~$t$ is not visible,
   a greedy path must trace an edge until the endpoint. If it is not
   possible, a local minimum must exist. }
  
\end{figure}

A greedy path is shown in Fig.~\ref{fig:greedy-path}. Note that such
paths are polylines. The way greedy paths are defined resembles
compass routing~\cite{kranakis_compass_1999}.

\subsection{Greedily Routable Regions.}

Greedily Routable Regions were introduced by Tan and
Kermarrec~\cite{tk-ggrlssn-2012} as follows.

\begin{definition}[\hspace{0.1mm}\cite{tk-ggrlssn-2012}]
  A polygonal region~$\mathcal P$ is a \emph{greedily routable region
    (GRR)}, if for any two points~$s, t \in \mathcal P$, $s \neq t$,
  point~$s$ can always move along a straight-line segment
  within~$\mathcal P$ to some point~$s'$ such that $|s't| < |st|$.
\label{def:grr}
\end{definition}

Next we show that~$\mathcal P$ is a GRR if and only if every pair of
points in~$\mathcal P$ is connected by a greedy path. Therefore,
Definition~\ref{def:grr} is equivalent to the one used in the
abstract. We shall show that the following procedure produces a greedy
path inside a GRR.

\bigskip

\makeatletter
\renewcommand{\@algocf@capt@plain}{above}
\makeatother
\begin{algorithm}[H]
\SetAlgorithmName{Procedure}{procedure}{List of Procedures}
  Set~$p = s$.
 
  If~$t$ is visible from $p$, move $p$ to $t$ and finish the
  procedure.  

  Move~$p$ to the first intersection of $pt$
  and~$\partial \mathcal P$. (Note that $p$ itself may be the first
  intersection.)

  If~$p$ is in the interior of a boundary edge~$v_1 v_2$, consider the
  angle between~$\vec {p v_i}$ and~$\vec {p t}$, $i = 1,2$. Let~$v_i$
  be the vertex minimizing~$\vangle {\vec {p v_i}}{\vec {p t}}$,
  $i=1,2$ (break ties arbitrarily). If~$v_i$ is the closest point
  to~$t$ on the segment~$p v_i$, move~$p$ to~$v_i$ and return to
  Step~2, otherwise, return failure.

  If~$p$ coincides with the vertex~$v_2$ incident to boundary
  edges~$v_1 v_2$ and~$v_2 v_3$, consider the angle
  between~$\vec {p v_i}$ and~$\vec {p t}$, $i = 1,3$. Let~$v_i$ be the
  vertex minimizing~$\vangle {\vec {p v_i}}{\vec {p t}}$, $i=1,3$
  (break ties arbitrarily). Again, if~$v_i$ is the closest point
  to~$t$ on the segment~$p v_i$, move~$p$ to~$v_i$ and return to
  Step~2, otherwise, return failure.

  \caption{Constructing a greedy $s$-$t$-path inside a GRR.}
  \label{alg:greedy-path}
\end{algorithm} 
\makeatletter
\renewcommand{\@algocf@capt@plain}{bottom}
\makeatother

\begin{lemma}
  \label{lem:grr-equiv}
  A polygonal region $\mathcal P$ is a GRR if and only if for
  every~$s,t \in \mathcal P$ there exists a
  greedy~$s$-$t$-path~$\rho \subseteq \mathcal
  P$. Procedure~\ref{alg:greedy-path} produces such a greedy path.
\end{lemma}
\begin{proof}
  First, consider~$s,t \in \mathcal P$ connected by a greedy
  $s$-$t$-path~$\rho$. Then $s,t$ satisfy the condition in
  Definition~\ref{def:grr} using the endpoint $s'$ of the first
  segment~$s s'$ of~$\rho$.

  Conversely, let~$\mathcal P$ be a GRR. Let~$s,t$ be two distinct
  points in $\mathcal P$, and consider a path~$\rho$ constructed by
  moving a point~$p$ from~$s$ to~$t$ according to
  Procedure~\ref{alg:greedy-path}. We consider the segments of~$\rho$
  iteratively and show that each of them would be taken by a greedy
  path.
  Since~$\mathcal P$ is a GRR, every point~$p \in \mathcal P$ can get
  closer to~$t$ by a linear movement. If all points on~$\ray{p}{t}$
  sufficiently close to~$p$ are in~$\mathcal P$, a greedy path would
  move along~$\ray{p}{t}$, until it hits~$\partial \mathcal P$.  This
  shows that Step~3 of the procedure traces a greedy path.

  Assume all points on~$\ray{p}{t}$ sufficiently close to~$p$ are not
  in~$\mathcal P$. Then,~$p$ is on~$\partial \mathcal
  P$. Let~$\vec{d_1}$ and~$\vec{d_2}$ be the two tangents in~$p$ to
  the paths that start at~$p$ and go along~$\partial \mathcal
  P$. Let~$\Lambda$ be the cone of directions spanned by~$\vec{d_1}$
  and~$\vec{d_2}$, such that~$\vec{pt} \notin
  \Lambda$. Then,~$\Lambda$ contains the directions of all possible
  straight-line movements from~$p$. By Definition~\ref{def:grr}, for
  some direction~$\vec{d} \in \Lambda$, we have
  $\vangle{\vec{pt}}{\vec{d}} < 90\dg$. But then,
  $\min_{i=1,2} \vangle{\vec{pt}}{\vec{d_i}} \leq
  \vangle{\vec{pt}}{\vec{d}} < 90\dg$. Therefore, a greedy path would
  continue in the direction~$\vec{d_i}$, as does~$\rho$. Let~$v_i$ be
  the endpoint of the edge containing~$p$, such
  that~$\vec{p v_i} = \vec{d_i}$. Therefore, $\angle t p v_i <
  90\dg$. We must show that a greedy path is traced if~$p$
  follows~$\vec{d_i}$ until~$v_i$. We
  have~$\angle p v_i t \geq 90\dg$. Otherwise, the projection
  point~$x$ of~$t$ on the line through~$p v_i$ lies in the interior of
  the segment~$p v_i$ and is a local minimum with respect to the
  distance to~$t$, which is not possible in a GRR; see
  Fig.~\ref{fig:greedy-path-alg}. Therefore, when~$p$ moves in the
  direction~$\vec{d_i}$ towards~$v_i$, its distance to~$t$ decreases
  continuously, and the forward tangent always has the minimum
  possible angle with respect to the direction towards~$t$.  This
  shows that Steps~4 and~5 of the procedure trace a greedy path and
  never return failure.

  It follows that, when moving along~$\rho$, point~$p$ either moves
  directly to~$t$ or slides along a boundary edge until it reaches one
  of the endpoints. Therefore, point~$p$ never reenters an edge and
  must finally reach~$t$. The forward tangent on~$\rho$ always
  satisfies the condition of Definition~\ref{def:greedy-path},
  therefore,~$\rho$ is a greedy $s$-$t$-path.

\end{proof}

A \emph{decomposition} of a polygonal region~$\mathcal P$ is a
partition of~$\mathcal P$ into polygonal regions~$\mathcal P_i$ with
no holes, $i=1, \dots, k$, such
that~$\bigcup_{i=1}^k \mathcal P_i = \mathcal P$ and
no~$\mathcal P_i$, $\mathcal P_j$ with $i \ne j$ share an interior
point. Recall that GRRs have no holes.  A decomposition of
$\mathcal P$ is a \emph{GRR decomposition} if each
component~$\mathcal P_i$ is a GRR.
We shall use the terms \emph{GRR decomposition} and \emph{GRR
  partition} interchangeably. 
Using the concept of a \emph{conflict relationship} between edges of a
polygonal region (see Fig.~\ref{fig:conflict}), Tan and Kermarrec give
a convenient characterization of GRRs.

\begin{definition}[Normal ray\hspace{0mm}]
  Let~$\mathcal P$ be a polygonal region, $e = uv$ a boundary edge
  and~$p$ an interior point of~$uv$. Let~$\nray{uv}{p}$ denote the ray
  with origin in~$p$ orthogonal to~$uv$, such that all points on this
  ray sufficiently close to~$p$ are not in the interior
  of~$\mathcal P$.
\end{definition}

  \begin{figure}[tb]
    \hfill
\subfloat[]{\includegraphics[scale=0.9,page=2]{fig/grr.pdf}\label{fig:heuristic-bad}}
    \hfill
\subfloat[]{\includegraphics[scale=0.9,page=1]{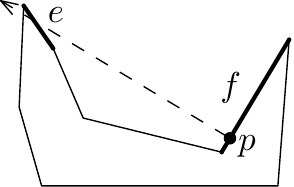}\label{fig:conflict}}
\hfill\null
\caption{\protect\subref{fig:heuristic-bad}~The heuristic
  in~\cite{tk-ggrlssn-2012} splits a non-greedy region by a bisector
  at a maximum reflex angle. If the splits are chosen in order of
  their index, seven regions are created, although two is minimum
  (split only at~6). \protect\subref{fig:conflict}~Normal ray~$\nray{f}{p}$ and
  a pair of conflicting edges~$e$,~$f$.}
  \end{figure}

  We restate the definition of conflicting edges
  from~\cite{tk-ggrlssn-2012}.

\begin{definition}[Conflicting edges of a polygonal region]
  Let~$e$ and~$f$ be two edges of a polygonal region~$\mathcal P$.
  If for some point~$p$ in the interior of~$e$, $\nray{e}{p}$
  intersects~$f$, then~$e$ \emph{conflicts} with~$f$.
\end{definition}
A polygonal region is a GRR if and only if it has no pair of
conflicting edges;~\cite[Theorem~1]{tk-ggrlssn-2012}.
Furthermore, GRRs are known to have no holes.

Now consider a plane straight-line drawing~$\Gamma$ of a
graph~$G=(V,E)$. We identify the edges of~$G$ with the corresponding
line segments of~$\Gamma$ and the vertices of~$G$ with the
corresponding points.  Plane straight-line drawings can be considered
as infinitely thin polygonal regions. The routing happens along the
edges of~$\Gamma$, and we define GRRs for graph drawings as follows.
\begin{definition}[GRRs for plane straight-line drawings]
  A plane straight-line graph drawing~$\Gamma$ is a GRR if for any two
  points~$s \neq t$ on~$\Gamma$ there exists a point~$s'$ on an edge
  that also contains~$s$, such that~$|s't| < |st|$.
\end{definition}
Note that for an interior point~$p$ of an edge~$e$ of~$\Gamma$ there
exist two normal rays at~$p$ with opposite
directions. Let~$\normal{e}{p}$ denote the normal line to~$e$
at~$p$. We define conflicting edges of~$\Gamma$ as follows.

\begin{definition}[Conflicting edges of a plane straight-line drawing]
  Let~$e$ and~$f$ be two edges of a plane straight-line
  drawing~$\Gamma$.  If for some point~$p$ in the interior of~$e$,
  $\normal{e}{p}$ intersects~$f$, then~$e$ \emph{conflicts}
  with~$f$.
\end{definition}

Assume $\normal{e}{s}$ for an interior point $s$ on an edge~$e$
of~$\Gamma$ crosses another edge~$f$ in point~$t$. Then, any
movement along~$e$ starting from~$s$ increases the distance
to~$t$. We call such edges \emph{conflicting}. It is easy to see
that~$\Gamma$ is a GRR if it contains no pair of conflicting
edges. Obviously, such a drawing~$\Gamma$ contains no cycles.  In
fact, a straight-line drawing of a tree is \ic if and only if it has
no conflicting edges~\cite{acglp-sag-12}, which implies the following
lemma.

\begin{lemma}
  \label{lem:tree-grr}  
  The following two properties are equivalent for a straight-line
  drawing~$\Gamma$ to be a GRR.
  \begin{compactenum}[1)]
    \item $\Gamma$ is connected and has no conflicting edges; 
    \item $\Gamma$ is an \ic drawing of a tree.
  \end{compactenum}
\end{lemma}

Since every individual edge in a straight-line drawing is a GRR, the following
observation can be made on the worst-case size of a minimum GRR
partition.
\begin{observation}
  A plane straight-line drawing~$\Gamma$ of graph~$G=(V,E)$, $|E|=m$,
  has a GRR decomposition of size~$m$.
\end{observation}

Therefore, if~$G$ is a tree, the drawing~$\Gamma$ has a GRR partition
of size~$n-1$ for~$n=|V|$.

\subsection{Splitting graph drawings at non-vertices.}\label{sec:split}

Note that in a GRR partition of a plane straight-line drawing~$\Gamma$
of a graph~$G=(V,E)$, an edge~$e \in E$ does not necessarily lie in
\emph{one} GRR. Pieces of the same edge can be part of different
GRRs. Allowing splitting edges at intermediate points might result in
smaller GRR partitions; see Fig.~\ref{fig:split-nonvertex}. In this
section, we discuss splitting~$\Gamma$ at non-vertices. We will show
that there are only a discrete set of $O(n^2)$ points where we might
need to split edges.

\begin{figure}[tb]
  \subfloat[]{\includegraphics[page=1]{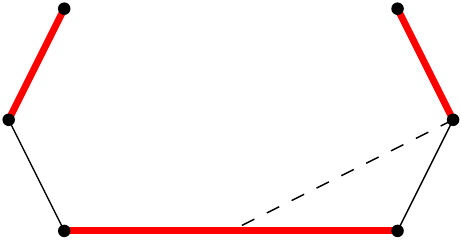}
    \label{fig:split-nonvertex:1}} \hfill
  \subfloat[]{\includegraphics[page=2]{fig/grr-partition-split.pdf}
    \label{fig:split-nonvertex:2}}
  \caption{Splitting at non-vertices results in a smaller
    partition. \protect\subref{fig:split-nonvertex:1}~No pair of the
    thick red edges can be in the same GRR. Therefore, if no edge
    splits are allowed, every GRR partition has size at
    least~3. \protect\subref{fig:split-nonvertex:2}~Splitting the
    longest edge results in a GRR partition of size~2.}
  \label{fig:split-nonvertex}
\end{figure}

\begin{definition}[Subdivided drawing~$\Gamma_s$]
  Let~$\Gamma_s$ be the drawing created by subdividing edges
  of~$\Gamma$ as follows. For every pair of original
  edges~$u_1 u_2, u_3 u_4 \in E$, let~$\ell_i$ be the normal
  to~$u_1 u_2$ at~$u_i$, $i=1,2$. If~$\ell_i$ intersects~$u_3 u_4$, we
  subdivide~$u_3 u_4$ at the intersection.
\end{definition}

Since we consider only the original edges of~$\Gamma$, the
subdivision~$\Gamma_s$ has $O(n^2)$ vertices.

\begin{lemma}
  Any GRR decomposition of~$\Gamma$ with potential edge splits can be
  transformed into a GRR decomposition of~$\Gamma_s$ in which no edge
  of~$\Gamma_s$ is split, such that the size of the decomposition does
  not increase.
  \label{lem:split}
\end{lemma}
\begin{proof}
  Consider edge~$u v$ of the subdivision~$\Gamma_s$, a point~$x$ in
  its interior and assume an \ic component~$C$ (green in
  Fig.~\ref{fig:trees:splits}) contains~$vx$, but not~$ux$. We claim
  that we can reassign~$ux$ to~$C$.  Note that iterative application of this
  claim implies the lemma.

\begin{figure}[tb]
  \hfill
\subfloat{
  \includegraphics[page=2]{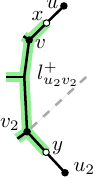}} 
\hfill\null
\caption{Proof of Lemma~\ref{lem:split}. Segment~$u x$ can be added to
  the thick green GRR~$C$, such that the entire edge~$u v$
  of~$\Gamma_s$ is in one GRR.}
\label{fig:trees:splits}
\end{figure}

For points $p,q \in \mathbb R^2, p \neq q$, let~$\hp{p}{q}$ denote the
halfplane not containing $p$ bounded by the line through~$q$
orthogonal to the segment~$pq$. Note that if segment~$p q$ is on the
path from vertex~$p$ to vertex~$r$ in an \ic tree drawing then~$r \in \hp{p}{q}$~\cite{acglp-sag-12}.

Let $u_2 v_2$ be an original edge of~$\Gamma$ such that~$v_2$ is
in~$C$, as well as a subsegment~$y v_2$ of~$u_2 v_2$ with a non-zero
length containing~$v_2$. Since segment~$y v_2$ is on the $y$-$v$-path 
in~$C$, the halfplane $\hp{u_2}{v_2} = \hp{y}{v_2}$ contains~$v$, and its boundary
does not cross~$uv$ by the construction of $\Gamma_s$.  Thus,
$\hp{u_2}{v_2}$ contains~$uv$. In this way, we have shown that no
normal ray of an edge of~$C$ crosses~$uv$.

  Furthermore, $\hp{u}{v} = \hp{x}{v}$. Since~$C - xv$ lies entirely
  in~$\hp{x}{v} = \hp{u}{v}$, this shows that no normal of~$uv$
  crosses another edge of~$C$. It follows that the union of~$C$
  and~$ux$ contains no conflicting edges and, therefore, is \ic by
  Lemma~\ref{lem:tree-grr}.

  Finally, removing~$ux$ from the component~$C'$ containing it doesn't
  disconnect them, since no edge or edge part is attached to~$x$ (or
  an interior point of~$ux$). Since~$C' - ux$ is connected and~$C'$ is
  a GRR, $C' - ux$ is also a GRR.
\end{proof}

\subsection{Types of GRR contacts in plane straight-line graph
  drawings}

We distinguish the types of contacts that two GRRs can have in a GRR
partition of a plane straight-line graph drawing.

\begin{definition}[Proper, non-crossing and crossing contacts]
  Consider two drawings~$\Gamma_1$, $\Gamma_2$ of trees with the only
  common point~$p$.
  \begin{compactenum}[1)]
  \item $\Gamma_1$ and~$\Gamma_2$ have a \emph{proper contact} if~$p$
    is a leaf in at least one of them.
  \item $\Gamma_1$ and~$\Gamma_2$ have a \emph{non-crossing contact}
    if in the clockwise ordering of edges of~$\Gamma_1$ and~$\Gamma_2$
    incident to~$p$, all edges of~$\Gamma_1$ (and, thus, also
    of~$\Gamma_2$) appear consecutively.
  \item $\Gamma_1$ and~$\Gamma_2$ are~\emph{crossing} or have a
    \emph{crossing contact} if in the clockwise ordering of edges
    of~$\Gamma_1$ and~$\Gamma_2$ incident to~$p$, edges of~$\Gamma_1$
    (and, thus, also of~$\Gamma_2$) appear non-consecutively.
\end{compactenum}
\label{def:contacts}
\end{definition}

The first part of Definition~\ref{def:contacts}  allows GRRs to only
have contacts as shown in Fig.~\ref{fig:trees:proper} and forbids contacts
as shown in
Fig.~\ref{fig:trees:non-proper:non-crossing},~\ref{fig:trees:non-proper:crossing}. The
second part allows contacts as those in
Fig.~\ref{fig:trees:non-proper:non-crossing}, but forbids the contacts
in Fig.~\ref{fig:trees:non-proper:crossing}.

Note that a contact of two trees~$\Gamma_1, \Gamma_2$ with a single
common point~$p$ is either crossing or non-crossing. Moreover, if the contact of~$\Gamma_1$ and~$\Gamma_2$ is proper, then it is
necessarily non-crossing, since for a proper contact,~$\Gamma_1$
or~$\Gamma_2$ has only one edge incident to~$p$, therefore, all edges
of~$\Gamma_1$ and of~$\Gamma_2$ appear consecutively around~$p$.

We shall show that for trees, restricting ourselves to GRR
decompositions with only non-crossing contacts makes the otherwise
\NP-complete problem of finding a minimum GRR partition solvable in
polynomial time.

\begin{figure}[tb]
  \hfill
\subfloat[]{
  \includegraphics[page=3]{fig/non-proper-contacts.pdf}\label{fig:trees:proper}}
  \hfill
\subfloat[]{
  \includegraphics[page=1]{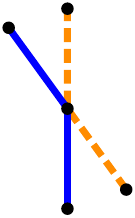}\label{fig:trees:non-proper:non-crossing}}
  \hfill 
\subfloat[]{
  \includegraphics[page=2]{fig/non-proper-contacts.pdf}\label{fig:trees:non-proper:crossing}}
\hfill\null
\caption{\protect\subref{fig:trees:proper}~Proper GRR contact;
  \protect\subref{fig:trees:non-proper:non-crossing}~non-crossing contact which is not proper and
  \protect\subref{fig:trees:non-proper:crossing}~crossing contact.
  \label{fig:trees:contacts}}
\end{figure}

\section{NP-completeness for graphs with cycles}

We show that finding a minimum decomposition of a plane straight-line
drawing~$\Gamma$ into \ic trees is \NP-hard. This extends the
\NP-hardness result by Tan and Kermarrec~\cite{tk-ggrlssn-2012} for
minimum GRR decompositions of polygonal regions with holes to plane
straight-line drawings.

Note that in the graph drawings used for our proof, all GRRs will
have \emph{proper contacts}; see
Definition~\ref{def:contacts}. Moreover, the graph drawings can be
turned into thin polygonal regions in a natural way by making them
slightly ``thicker'', and the proof can be reused as another proof for
the \NP-hardness result in~\cite{tk-ggrlssn-2012}.

Both our \NP-hardness proof and the proof in~\cite{tk-ggrlssn-2012}
are reductions from the \NP-complete problem \textsc{Planar
  3SAT}~\cite{l-pftu-82}.  Recall that a Boolean 3SAT formula
$\varphi$ is called \emph{planar}, if the corresponding variable
clause graph $G_\varphi$ having a vertex for each variable and for
each clause and an edge for each occurrence of a variable (or its
negation) in a clause is a planar graph.  In fact, $G_\varphi$ can be
drawn in the plane such that all variable vertices are aligned on a
vertical line and all clause vertices lie either to the left or to the
right of this line and connect to the variables via E- or
$\exists$-shapes~\cite{kr-pcr-92}; see Fig.~\ref{fig:planar3sat}.

\begin{figure}[tb]
  \centering
  \includegraphics[]{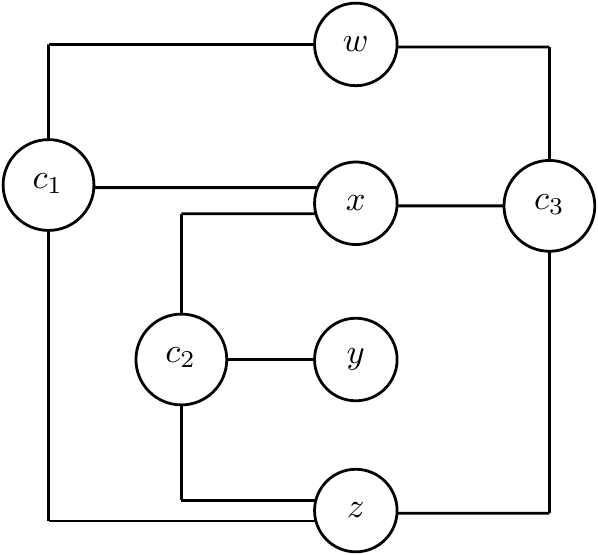}
  \caption{An orthogonal graph drawing of the variable-clause
    graph~$G_\varphi$ for a planar 3SAT formula
    $\phi = (w \vee x \vee z) \wedge~(\bar x \vee y \vee \bar z)
    \wedge~(\bar w \vee \bar x \vee \bar z)$.}
  \label{fig:planar3sat}
\end{figure}

The basic idea of the gadget proof is as follows. Using a number of
building blocks, or \emph{gadgets}, we construct a plane straight-line
drawing~$\Gamma_\varphi$, whose geometry mimics the variable-clause
graph~$G_\varphi$ drawn as described above. We
construct~$\Gamma_\varphi$ in a way such that its minimum GRR
decompositions are in correspondence with the truth assignments of the
\textsc{Planar 3SAT} formula~$\varphi$.

The variable gadgets in~\cite{tk-ggrlssn-2012} are cycles formed by
T-shaped polygons which can be made arbitrarily thin. Thus, in the
case of plane straight-line drawings we can use very similar variable
gadgets (see Fig.~\ref{fig:cycles:variable}). The clause gadgets
in~\cite{tk-ggrlssn-2012}, however, are squares, at which three
variable cycles meet. This construction cannot be adapted for
straight-line plane drawings, and we have to construct a significantly
different clause gadget; see Fig.~\ref{fig:cycles:clause}.

\begin{figure}[tb]
  \hfill
  \subfloat[]{\includegraphics[scale=0.9,page=1]{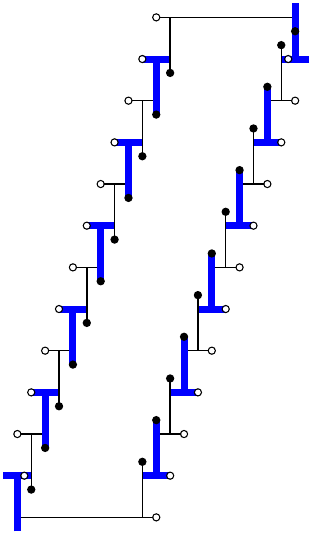}\label{fig:cycles:var-true}}
  \hfill
  \subfloat[]{\includegraphics[scale=0.9,page=2]{fig/variable-1.pdf}\label{fig:cycles:var-false}}
  \hfill
  \subfloat[]{\includegraphics[scale=0.9,page=3]{fig/variable-1.pdf}\label{fig:cycles:var-arm-stubs}}
  \hfill\null
  \caption{Variable gadget and the two possibilities to pair vertical
    and horizontal segments to make GRRs:
    \protect\subref{fig:cycles:var-true}~\emph{true} variable state:
    $\top$-shapes and $\bot$-shapes;
    \protect\subref{fig:cycles:var-false}~\emph{false} variable state:
    $\dashv$-shapes and~$\vdash$-shapes.
    \protect\subref{fig:cycles:var-arm-stubs} Extending the variable
    gadgets to create the upper, middle and lower arm gadgets by
    substituting T-shapes of the variable gadget.}
  \label{fig:cycles:variable}
\end{figure} 

\begin{figure}[tb]
  \hfill
  \subfloat[]{\includegraphics[scale=0.9,page=1]{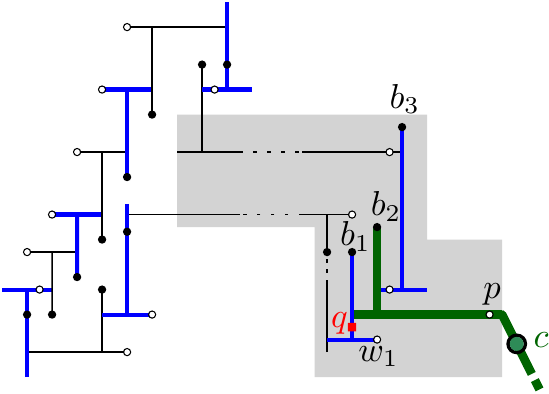}\label{fig:cycles:arm-pos-true}}
  \hfill
  \subfloat[]{\includegraphics[scale=0.9,page=2]{fig/arm-1.pdf}\label{fig:cycles:arm-pos-false}}
  \hfill\null
  \caption{Variable gadget with a right upper positive arm (shaded
    region). \protect\subref{fig:cycles:arm-pos-true}~\emph{true} and
    \protect\subref{fig:cycles:arm-pos-false}~\emph{false} states.}
  \label{fig:cycles:arm}
\end{figure} 

\begin{figure}[tb]
  \subfloat[]{\includegraphics[scale=1.05,page=1]{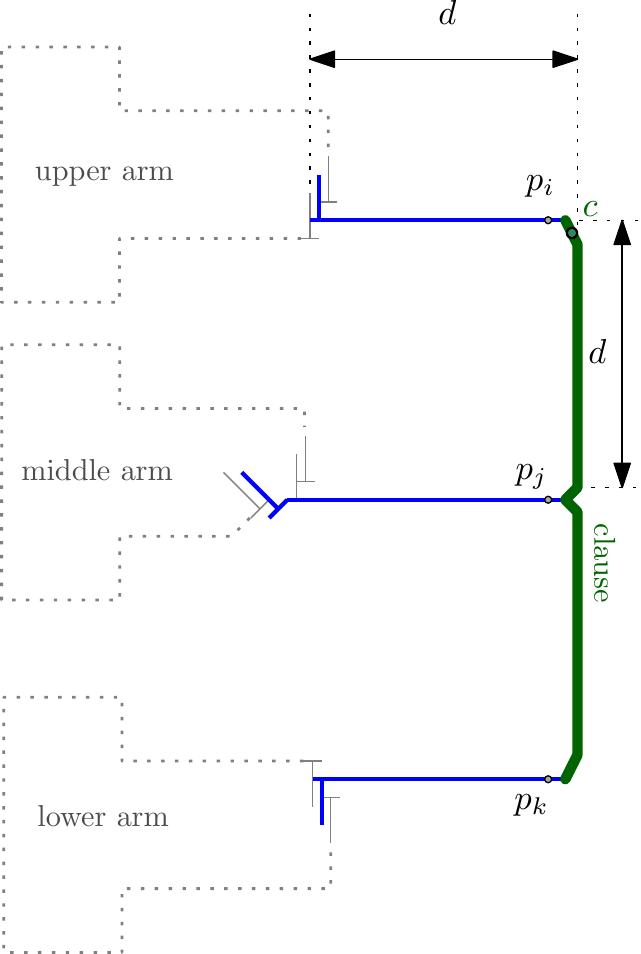}\label{fig:cycles:clause:3true}}
  \hfill
  \subfloat[]{\includegraphics[scale=1.05,page=2]{fig/clause.pdf}\label{fig:cycles:clause:3false}}
  \caption{Clause gadget (thick green). (a) \emph{true} and (b)
    \emph{false} state of the involved literals.}
  \label{fig:cycles:clause}
\end{figure}

\begin{figure}[tb]
  \subfloat[]{\includegraphics[scale=1,page=4]{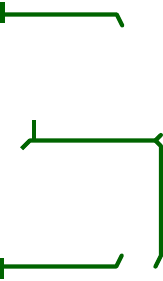}\label{fig:cycles:clause:merged:false}}
  \hfill
  \subfloat[]{\includegraphics[scale=1,page=3]{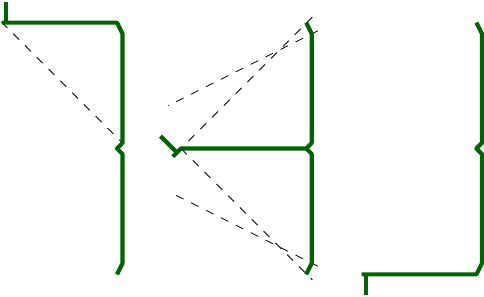}\label{fig:cycles:clause:merged:true}}
  \caption{Merging the clause gadget with GRRs from the arm
    loops. (a)~None of the three components is a GRR. (b)~ All three
    components are GRRs; see the dashed normals. }
  \label{fig:cycles:clause-merged}
\end{figure}

We define a variable gadget as a cycle of alternating vertical and
horizontal segments. The tip of each segment touches an interior point
of the next segment. We can join pairs of consecutive segments into
a GRR by assigning each vertical segment either to the next
or to the previous horizontal segment on the cycle. In this way, the
variable loop is partitioned either in $\top$-shapes and $\bot$-shapes
or in~$\dashv$-shapes and~$\vdash$-shapes; see
Fig.~\ref{fig:cycles:variable}.

Consider a variable gadget consisting of~$k$ T-shapes as shown in
Fig.~\ref{fig:cycles:variable}.  On each T-shape we place one black
and one white point as shown in the figure. The points are placed in
such a way that neither two black points nor two white points can be
in one \ic component. Thus, a minimum GRR decomposition of a variable
gadget contains at least~$k$ components. If it contains exactly~$k$
components, then each component must contain one black and one white
point, and there are exactly two possibilities. Each black point has
exactly two white points it can share a GRR with, and once one pairing
is picked, it fixes all the remaining pairings. The corresponding
possibilities are shown in Fig.~\ref{fig:cycles:var-true}
and~\ref{fig:cycles:var-false} and will be used to encode the values
\emph{true} and \emph{false}, respectively. For the pairing of the
black and white points corresponding to the true state, the variable
loop can be partitioned in $\top$-shapes and $\bot$-shapes, and for
the pairing corresponding to the false state, it can be partitioned
in~$\dashv$-shapes and~$\vdash$-shapes.

To pass the truth assignment of a variable to a clause it is part of,
we use \emph{arm} gadgets. Arm gadgets are extensions of the variable
gadget.
%
%
To add an arm gadget to the variable, we substitute several $\top$- or
$\bot$-shapes from the variable loop by a more complicated
structure. Fig.~\ref{fig:cycles:var-arm-stubs} shows such extensions
for all arm types pointing to the right, the other case is
symmetric. In this way, for a variable, we can create as many arms as
necessary.
 Each variable loop will have one arm extension for each
occurrence of the corresponding variable in a clause in~$\varphi$. The
working principle for the arm gadgets is the same as for the variable
gadgets. The drawing created by the variable cycle and the arm
extensions (the \emph{variable-arm loop}) will once again contain
distinguished black and white points, such that only one black and one
white point can be in a GRR. However, for variable-arm loops, the
cycles formed by segments of varying orientation are more complicated
than the loop in Fig.~\ref{fig:cycles:variable}. For example, for some
arm types we use segments of slopes~$\pm 1$ in addition to vertical
and horizontal segments.

In total twelve variations of the arm gadget will be used, depending
on the position of the literal in the clause, the position of the
clause, and whether the literal is negated or not.  Since in
$G_\varphi$ each clause $c$ connects to three variables, we denote
these variables or literals as the \emph{upper}, \emph{middle}, and
\emph{lower} variables of $c$ depending on the order of the three
edges incident to~$c$ in the one-bend orthogonal drawing
of~$G_\varphi$ used by Knuth and Raghunathan~\cite{kr-pcr-92}; see
Fig.~\ref{fig:planar3sat}.
Similarly, an arm of $c$ is called an \emph{upper}, \emph{middle}, or
\emph{lower} arm if it belongs to a literal of the same type in~$c$.
An arm is called a \emph{right} (resp. \emph{left}) arm if it belongs
to a clause that lies to the right (resp. to the left) of the vertical
variable line.
Finally, an arm of $c$ is \emph{positive} if the corresponding literal
is positive in $c$ and it is \emph{negative} otherwise.

The basic principle of operation of any arm gadget is the same; as an
example consider the right upper positive arm in
Fig.~\ref{fig:cycles:arm}.
Figures~\ref{fig:cycles:arm-2}, \ref{fig:cycles:arm-3},
\ref{fig:cycles:arm-inv} and the proof of
Property~\ref{lem:cycles:arms} cover the remaining arm types.

The positive and the negative arms are differentiated by an additional
structure that switches the pairing of the black and white points
close to the part of the arm that touches the clause gadget; for
example, compare Fig.~\ref{fig:cycles:arm-pos-false}
and~\ref{fig:cycles:arm-inv-1}. By this inversion, for a fixed truth
assignment of the variable, the~$\top$- and~$\bot$-shapes next to the
clause are turned into~$\vdash$- and~$\dashv$-shapes, and vice
versa. In this way, the inverted truth assignment of the corresponding
variable is passed to the clause.

Note that each arm can be arbitrarily extended both horizontally and
vertically to reach the required point of its clause gadget. We select
again black and white points (also called \emph{distinguished} points)
on the line segments of the arm gadget.

The \emph{clause gadget} (the thickest green polyline in
Fig.~\ref{fig:cycles:clause}, partly drawn in
Fig.~\ref{fig:cycles:arm}) is a polyline which consists of six
segments. The first segment has slope~2, the second is vertical, the
third has slope~$-1$, the fourth has slope~1, the fifth is vertical,
and the sixth has slope~$-2$. Each clause gadget connects to the long
horizontal segments of the arms of three variable gadgets. The three
connecting points of the clause gadget are the start and end of the
polyline as well as its center, which is the common point of the two
segments with slopes~$\pm 1$.

We shall prove the following property which is crucial for our
construction.

\begin{property}
  \begin{compactenum}
  \item Consider a drawing~$\Gamma_i$ of a variable gadget together
    with all of its arms. Then, neither two black nor two white points
    on~$\Gamma_i$ can be in one GRR. In a minimum GRR decomposition
    of~$\Gamma_i$, each component has one black and one white point,
    and exactly two such pairings of points are possible, one for each
    truth assignment.
  \item Consider two such drawings~$\Gamma_i$, $\Gamma_j$ for two
    different variables. Then, no distinguished point of~$\Gamma_i$
    can be in the same GRR as a distinguished point of~$\Gamma_j$.
  \end{compactenum}
  \label{lem:np:black-white-points}
\end{property}

\begin{proof}
Part~(1) of Property~\ref{lem:np:black-white-points} extends the same
property that we already showed for variable gadgets without arms to
the case including all arms. It is an immediate consequence of the way
we constructed the arm gadgets and placed the distinguished points;
see Figures~\ref{fig:cycles:arm}, \ref{fig:cycles:arm-2},
\ref{fig:cycles:arm-3}, \ref{fig:cycles:arm-inv}.

Part~(2) follows from the way the arms are connected by a clause,
i.e., in Fig.~\ref{fig:cycles:clause} no pair of points from~$p_i$,
$p_j$, $p_k$ can be in the same GRR, since the three points lie on
three horizontal segments and are vertically collinear.
\end{proof}

The clause gadget is connected to the arm by a horizontal segment with
a distinguished point~$p$ on its end, which is either black or white
depending on the arm type. Each clause has one special point~$c$
chosen as shown in Fig.~\ref{fig:cycles:clause}.

We show that~$c$ and~$p$ can be in the same GRR in a minimum GRR
decomposition if and only if the variable gadget containing~$p$ is in
the state that satisfies the clause.

\begin{property}
  \begin{compactenum}
  \item In a minimum GRR decomposition, the special point~$c$ of a
    clause gadget can share a GRR with a black or white point of an
    arm gadget if and only if the corresponding literal is in the
    \emph{true} state.
  \item If a variable assignment satisfies a clause, then its entire
    clause gadget can be contained in a GRR of an arm corresponding to
    a \emph{true} literal.
\end{compactenum}
  \label{lem:cycles:arms} 
\end{property}

\begin{proof}
  For each arm gadget we select a special \emph{red} point~$q$; see
  Fig.~\ref{fig:cycles:arm}. Point~$q$ is neither white nor black.
  By Property~\ref{lem:np:black-white-points}, in a minimum GRR
  decomposition, point~$q$ must be in a GRR together with one black
  and one white point.

  For the various arm types, if points~$q$ and~$p$ are in the same
  GRR, we shall show that this GRR cannot contain the entire clause
  gadget and, in particular, cannot contain point~$c$. This is
  illustrated in Fig.~\ref{fig:cycles:clause:merged:false}.

  Furthermore, we shall show that if the literal is in the \emph{true}
  state, then points~$p$ and~$q$ are in different GRRs, and the GRR
  containing~$p$ can be merged with the entire clause gadget,
  including $c$.
  For example, in Fig.~\ref{fig:cycles:clause:3true}, each variable is
  in a state that satisfies the clause. The lengths of the thick
  segments are chosen such that each thick blue component can be
  merged with the clause gadget (thickest green) into a single GRR, as
  shown in Fig.~\ref{fig:cycles:clause:merged:true}.

 \medskip

$i)$~We first show the lemma for a positive right upper arm.  We use the
notation from Fig.~\ref{fig:cycles:arm} to refer to the distinguished
points. In the \emph{true} state of the variable (see
Fig.~\ref{fig:cycles:arm-pos-true}), points~$w_1$, $b_1$ and~$q$ are
in the same GRR. Points~$b_2$ and~$p$ are in another GRR (e.g., the
thickest green one in Fig.~\ref{fig:cycles:arm}) which can contain the
distinguished point~$c$ of the clause.

In the \emph{false} state of the variable (see
Fig.~\ref{fig:cycles:arm-pos-false}), the points~$b_1$ and~$p$ are in
the same GRR.  Moreover, point~$q$ can share a GRR with exactly one
point from~$b_1$, $b_2$ or~$b_3$. But if~$q$ were with~$b_2$ or~$b_3$,
then~$b_1$ would be disconnected from any white point, a contradiction
to the minimality of the decomposition. Thus, points~$q$, $b_1$
and~$p$ are in the same GRR, which cannot contain a point of the
clause.
	
  \begin{figure}[htb] \hfill
    \subfloat[]{\includegraphics[page=1]{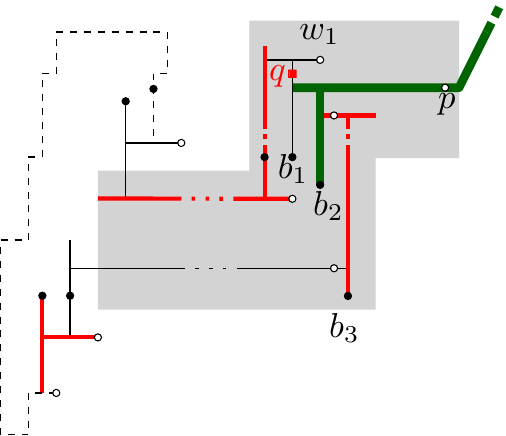}\label{fig:cycles:arm-neg-false-2}}
    \hfill
    \subfloat[]{\includegraphics[page=2]{fig/arm-2.pdf}\label{fig:cycles:arm-neg-true-2}}
    \hfill\null
    \caption{Right lower negative arm
      gadget. \protect\subref{fig:cycles:arm-neg-false-2}~\emph{false}
      and \protect\subref{fig:cycles:arm-neg-true-2}~\emph{true}
      variable state. Thin dashed lines indicate that the variable-arm
      loop continues.}
    \label{fig:cycles:arm-2}
  \end{figure}
	
  \medskip
	
  $ii)$~We now show the lemma for a negative right lower arm. We use
  the notation from Fig.~\ref{fig:cycles:arm-2}. In the \emph{false}
  state of the variable (which corresponds to the \emph{true} state of
  the considered literal), points~$w_1$, $b_1$ and~$q$ are in the same
  GRR; see Fig.~\ref{fig:cycles:arm-neg-false-2} Points~$b_2$ and~$p$
  are in another GRR (e.g., the very thick green one in
  Fig.~\ref{fig:cycles:arm}) which can contain the entire clause; see
  the lower arm in Fig.~\ref{fig:cycles:clause} and the corresponding
merged component in Fig.~\ref{fig:cycles:clause:merged:true}.

  Now consider a \emph{true} state of the variable; see
  Fig.~\ref{fig:cycles:arm-neg-true-2}. Point~$q$ shares a GRR with
  exactly one point from~$b_1$, $b_2$ or~$b_3$. If~$q$ is with~$b_2$
  or~$b_3$, then~$b_1$ is disconnected from any white point, a
  contradiction to the minimality of the decomposition. Thus,
  points~$q$, $b_1$ and~$p$ are in the same GRR, which cannot contain
  a point of the clause.

  \medskip

  $iii)$~Next, consider a positive right middle arm; see
  Fig.~\ref{fig:cycles:arm-3}. We identify points~$p$ and~$b_1$.
  Point~$b_1$ is either with~$w_0$ (\emph{true} state of the variable)
  or~$w_1$ (\emph{false} state of the variable).

  In the \emph{true} state, points~$b_1$ and~$w_0$ are in one GRR,
  which cannot contain~$q$. This GRR can be merged with the clause
  gadget; see Fig.~\ref{fig:cycles:arm-neg-true-3},
  \ref{fig:cycles:clause} and~\ref{fig:cycles:clause:merged:true}.

  In the \emph{false} state, points~$b_1$, $w_1$ and~$q$ are in one
  GRR, which cannot contain point~$c$ of the clause.

\begin{figure}[htb]
  \hfill
  \subfloat[]{\includegraphics[scale=1.0,page=1]{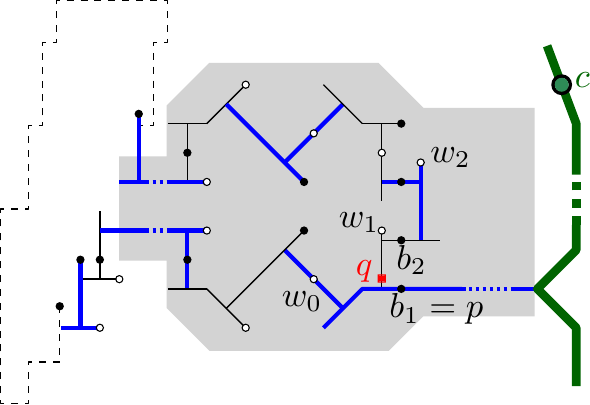}\label{fig:cycles:arm-neg-true-3}}
  \hfill
  \subfloat[]{\includegraphics[scale=1.0,page=2]{fig/arm-3.pdf}\label{fig:cycles:arm-neg-false-3}}
  \hfill\null
  \caption{Right positive middle arm
    gadget. \protect\subref{fig:cycles:arm-neg-true-3}~\emph{true} and
    \protect\subref{fig:cycles:arm-neg-false-3}~\emph{false} variable
    state.}
  \label{fig:cycles:arm-3}
\end{figure}

\medskip

$iv)$~To construct the negative right upper arm, the positive right
lower arm and the negative right middle arm, we invert the arm gadgets
constructed before. The inverted gadgets are shown in
Fig.~\ref{fig:cycles:arm-inv}. The proofs are analogous to the
respective non-inverted cases.

\medskip

$v)$~The left arms are constructed by mirroring.
\end{proof}

\begin{figure}[htb]
  \hfill
  \subfloat[]{\includegraphics[scale=0.80,page=1]{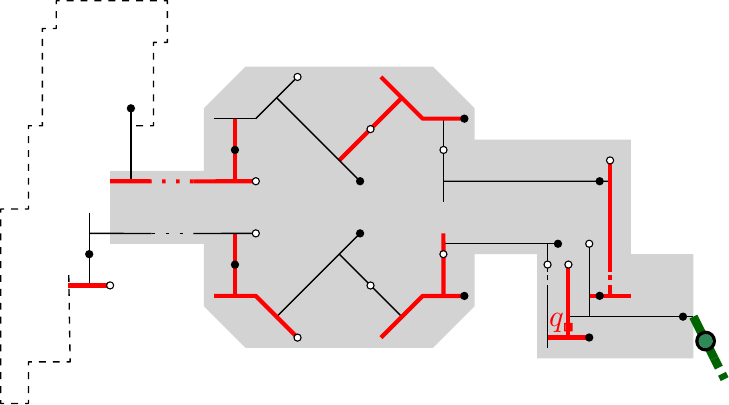}\label{fig:cycles:arm-inv-1}}
  \hfill
  \subfloat[]{\includegraphics[scale=0.80,page=2]{fig/arm-inv.pdf}\label{fig:cycles:arm-inv-2}}
  \hfill\null\\
  \centering \subfloat[]{\includegraphics[scale=1.0,page=3]{fig/arm-inv.pdf}\label{fig:cycles:arm-inv-3}}
  
  \caption{The remaining three right arms in the satisfying variable
    state. \protect\subref{fig:cycles:arm-inv-1}~negative right upper
    arm, \protect\subref{fig:cycles:arm-inv-2}~the positive right
    lower arm and \protect\subref{fig:cycles:arm-inv-3}~the negative
    right middle arm.}
  \label{fig:cycles:arm-inv}
\end{figure}

Finally, we can prove the \NP-hardness result by showing that any
satisfying truth assignment for a formula $\varphi$ yields a GRR
decomposition into a fixed number $k$ of GRRs, where $k$ is the total
number of black points in our construction. Likewise, using
Property~\ref{lem:np:black-white-points} and~\ref{lem:cycles:arms}, we
can show that any decomposition into~$k$ GRRs necessarily satisfies
each clause in $\varphi$.

\newcommand{\ThmNPhard}{For~$k \in \N$, deciding whether a plane
  straight-line drawing can be partitioned into~$k$ \ic components is
  \NP-complete.  }

\begin{theorem}
	\ThmNPhard
  \label{thm:cycles:nphard}  
\end{theorem}

\begin{proof}
  First, we show that the problem is in~\NP. Given a plane
  straight-line drawing~$\Gamma$, we construct its
  subdivision~$\Gamma_s$ as described in Section~\ref{sec:split}. By
  Lemma~\ref{lem:split}, it is sufficient to consider only partitions
  of edges in~$\Gamma_s$ into~$k$ components. To verify a positive
  instance, we non-deterministically guess the partition of the edges
  of~$\Gamma_s$ into~$k$ components. Testing if each component is a
  tree and if it is \ic can be done in polynomial time.

  Next, we show \NP-hardness.  Given a Planar 3SAT formula~$\varphi$,
  we construct a plane straight-line drawing $\Gamma_\varphi$ using
  the gadgets described above. It is easy to see that $\Gamma_\varphi$
  can be constructed on an integer grid of polynomial size and in
  polynomial time. Let $k$ be the number of black points produced by
  the construction. Note that~$k$ is $O(m+n)$, where~$n$ is the number
  of variables and~$m$ the number of clauses in~$\varphi$. We claim
  that~$\Gamma_\varphi$ can be decomposed in~$k$ GRRs if and only
  if~$\varphi$ is satisfiable.

  Consider a truth assignment of the variables
  satisfying~$\varphi$. We decompose each variable gadget and the
  attached arms as intended in our gadget design, which yields
  exactly~$k$ GRRs. By Property~\ref{lem:cycles:arms}, each
  clause gadget can be merged with the GRR of the arm of a literal
  which satisfies the clause. Therefore, we have~$k$ GRRs in total.

  Conversely, consider a decomposition of~$\Gamma_\varphi$ into~$k$
  GRRs. Then, each variable and the attached arms must be decomposed
  minimally and, by Property~\ref{lem:np:black-white-points}, must be
  either in the \emph{true} or in the \emph{false} state. Furthermore,
  each special point~$c$ of a clause must be in a component belonging
  to one of the arms of the clause. But then, the corresponding
  variable must satisfy the clause by
  Lemma~\ref{lem:cycles:arms}. This induces a satisfying variable
  assignment for~$\varphi$.
\end{proof}

\section{Trees}

In this section we consider \emph{greedy tree decompositions}, or
\gtds. For trees, greedy regions correspond to \ic drawings. Note that
\ic tree drawings are either subdivisions of~$K_{1,4}$, subdivisions
of the~\emph{windmill} graph (three caterpillars with maximum degree~3
attached at their ``tails'') or paths; see the characterization by
Alamdari et al.~\cite{acglp-sag-12}.

In the following, we consider a plane straight-line drawing~$\Gamma$
of a tree $T=(V,E)$, with~$|V|=n$. As before, we identify the tree
with its drawing, the vertices with the corresponding points and the
edges with the corresponding line segments. We want to partition it
into a minimum number of \ic subdrawings. In such a partition, each
pair of components shares at most one point.

Recall that a contact of two trees~$\Gamma_1, \Gamma_2$ with a single
common point~$p$ is either crossing or non-crossing; see
Definition~\ref{def:contacts}. Also, recall that proper contacts are
non-crossing.
Let~$\Pi_\textnormal{all}$ be the set of all GRR partitions of the
plane straight-line tree drawing~$\Gamma$. Let~$\Pi_{nc}$ be the set
of GRR partitions of~$\Gamma$, in which every pair of GRRs has a
non-crossing contact. Finally, let~$\Pi_{p}$ be the set of GRR
partitions of~$\Gamma$, in which every pair of GRRs has a proper
contact. It
holds:~$\Pi_{p} \subseteq \Pi_{nc} \subseteq
\Pi_\textnormal{all}$. For minimum
partitions~$\pi_{p}, \pi_{nc},\pi_\textnormal{all}$ from
$\Pi_{p}, \Pi_{nc}, \Pi_\textnormal{all}$, respectively, we have
$|\pi_\textnormal{all}| \leq |\pi_{nc}| \leq
|\pi_{p}|$.

We show that finding a minimum GTD of a plane straight-line tree
drawing is \NP-hard; see Section~\ref{sec:trees:npc}. In
Section~\ref{sec:trees:restricted-contacts}, we show that the problem
becomes polynomial if we consider GRR partitions in which GRRs
have only non-crossing contacts, i.e., partitions from~$\Pi_{nc}$. The
same holds if we only consider GRR partitions in which GRRs only have
proper contacts, i.e., partitions from~$\Pi_{p}$.

\subsection{NP-completeness}
\label{sec:trees:npc}

We show that if GRR crossings as in Definition~\ref{def:contacts}
 are allowed, deciding
whether a partition of given size exists is NP-complete.

The problem \textsc{Partition into Triangles (PIT)} has been shown to be
\NP-complete by {\'C}usti{\'c} et
al.~\cite[Proposition~5.1]{custic_geometric_2015} and will be
useful for our hardness proof.

\begin{problem}[PIT]
  Given a tripartite graph~$G=(V,E)$ with
  tripartition~$V=V_1 \cupdot V_2 \cupdot V_3$, 
  where~$|V_1| = |V_2| = |V_3| = q$. Does there exist a set~$T$ of~$q$
  triples in~$V_1 \times V_2 \times V_3$, such that every vertex
  in~$V$ occurs in exactly one triple and such that every triple
  induces a triangle in~$G$?
\end{problem}

It is easy to show that the following, similar problem
\textsc{Partition into Independent Triples (PIIT)} is \NP-complete as
well.

\begin{problem}[PIIT]
  Given a tripartite graph~$G=(V,E)$ with
  tripartition~$V=V_1 \cupdot V_2 \cupdot V_3$,
  where~$|V_1| = |V_2| = |V_3| = q$. Does there exist a set~$T$ of~$q$
  triples in~$V_1 \times V_2 \times V_3$, such that every vertex
  in~$V$ occurs in exactly one triple and such that no two vertices of
  a triple are connected by an edge in~$G$?
\end{problem}

\begin{lemma}
  PIIT is \NP-complete.
\end{lemma}
\begin{proof}
  It is easy to see that PIIT is in~\NP. For \NP-hardness, consider a
  graph~$G=(V,E)$ from an instance of~PIT. We construct~$G' = (V,E')$
  with
  $E' = \{ uv \mid uv \not \in E,~u \in V_i,~v \in V_j,~i \neq j
  \textnormal{ for } i,j=1,2,3 \}$. In this way, a triple
  from~$V_1 \times V_2 \times V_3$ induces a triangle in~$G$ if and
  only if it is independent in~$G'$. Therefore, PIT can be reduced to
  PIIT in polynomial time.
\end{proof}

We now show that deciding whether a GRR partition of a plane
straight-line tree drawing of given size exists is \NP-complete even
for subdivisions of a star.

\begin{theorem}
  Given a plane straight-line drawing~$\Gamma$ of a tree~$T=(V,E)$, which is a
  subdivision of a star with~$3q$ leaves, it is \NP-complete to decide
  whether~$\Gamma$ can be partitioned into~$q$ GRRs.
 \label{thm:grr:crossing-nphard}
\end{theorem}
\begin{proof}
  The proof that the problem is in \NP\xspace is analogous to the
  corresponding proof of Theorem~\ref{thm:cycles:nphard}.

  To prove \NP-hardness, we present a polynomial-time reduction from
  PIIT. Consider the tripartite graph~$G=(V,E)$ with
  tripartition~$V=V_1 \cup V_2 \cup V_3$ from an
  instance~$\Pi = (G,V_1,V_2,V_3,q)$ of PIIT,
  where~$|V_1| = |V_2| = |V_3| = q$. We may assume~$q \geq 3$. We show
  how to construct a plane straight-line drawing~$\Gamma$ of a
  subdivision of a star in polynomial time, such that~$\Gamma$ can be
  partitioned into~$q$ GRRs if and only if $\Pi$ is a yes-instance of
  PIIT. Figure~\ref{fig:grr:crossing-nphard:example} shows an example
  of such a construction for~$q=3$.

  \begin{figure}[tb]
    \centering
    \includegraphics[page=6]{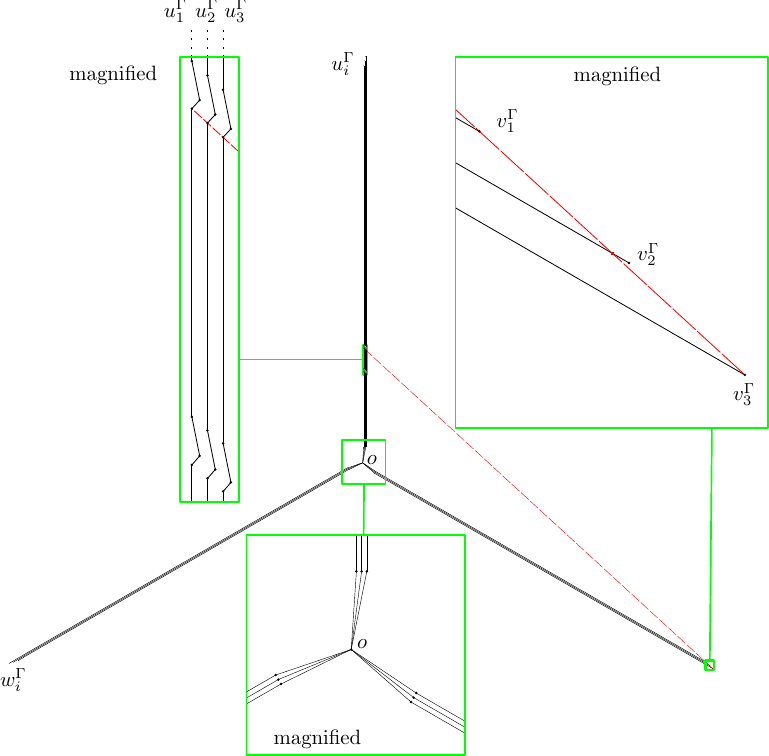}
    \caption{Reduction from a PIIT instance with~$q=3$ for the proof
      of Theorem~\ref{thm:grr:crossing-nphard}.}
    \label{fig:grr:crossing-nphard:example}
  \end{figure}

  We use the following basic ideas to construct the drawing~$\Gamma$.
  Let~$o$ be the center of~$\Gamma$.  Each vertex~$v$ of~$G$
  corresponds to a leaf vertex~$v^\Gamma$ of~$\Gamma$. The leaves
  of~$\Gamma$ are partitioned into three sets corresponding
  to~$V_1, V_2, V_3$. Consider a pair of vertices $u \in V_i$,
  $v \in V_j$. If~$i=j$, the angle that the~$u^\Gamma$-$v^\Gamma$ path
  has at point~$o$ in our construction is at
  most~$12\dg$. Therefore,~$u$ and~$v$ can not be in the same
  GRR. For~$i \neq j$, however, the angle that
  the~$u^\Gamma$-$v^\Gamma$ path has at point~$o$ is between~$106\dg$
  and~$134\dg$. We construct the~$o$-$u^\Gamma$ and~$o$-$v^\Gamma$
  paths in such a way that the~$u^\Gamma$-$v^\Gamma$ path is \ic if
  and only if edge~$u v$ is not in~$G$.

  The path from~$o$ to~$v^\Gamma$ takes a left turn of at most~$12\dg$
  and then continues as a straight line, except for at most~$q$
  \emph{dents}; see the left magnified part of
  Fig.~\ref{fig:grr:crossing-nphard:example}. Each dent is used to
  realize exactly one edge from~$G$. For a pair of vertices
  $u \in V_i$, $v \in V_j$, $j \equiv i+1~(\textrm{mod } 3)$ with
  edge~$uv$ in~$G$, the $o$-$u^\Gamma$ path has a dent with a normal
  crossing the $o$-$v^\Gamma$ path. Furthermore, no normal to this
  dent crosses the~$o$-$w^\Gamma$ path for any vertex
  $w \in V_j \cup V_k \setminus \{ v \}$,
  for~$k \equiv i+2~(\textrm{mod } 3)$. Consider the example in
  Fig.~\ref{fig:grr:crossing-nphard:example}.  Assume that there is an
  edge~$u_3 v_2$ in~$G$. Then, the~$o$-$u_3^\Gamma$ path has a dent
  whose normal (dashed red) crosses the $o$-$v_2^\Gamma$ path, but not
  the paths from~$o$ to~$v_1^\Gamma$, $v_3^\Gamma$, $w_1^\Gamma$,
  $w_2^\Gamma$ and~$w_3^\Gamma$.

  \begin{figure}[h!tb]
    \centering \subfloat[]{\includegraphics[page=2,
      scale=0.8]{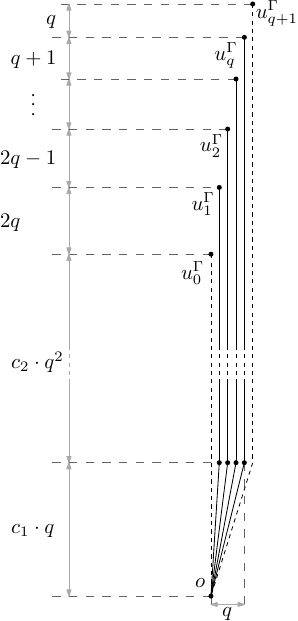}
      \label{fig:grr:crossing-nphard:constr:1}}\hfill
    \subfloat[]{\includegraphics[page=3,
      scale=0.8]{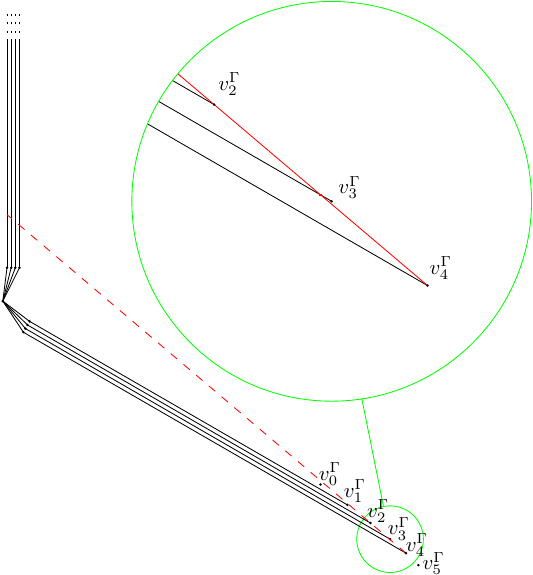}
      \label{fig:grr:crossing-nphard:constr:2}}\\
    \subfloat[]{\includegraphics[page=4]{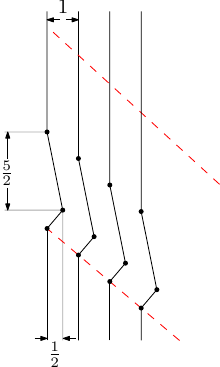}
      \label{fig:grr:crossing-nphard:constr:3}}\hfill
    \subfloat[]{\includegraphics[page=5]{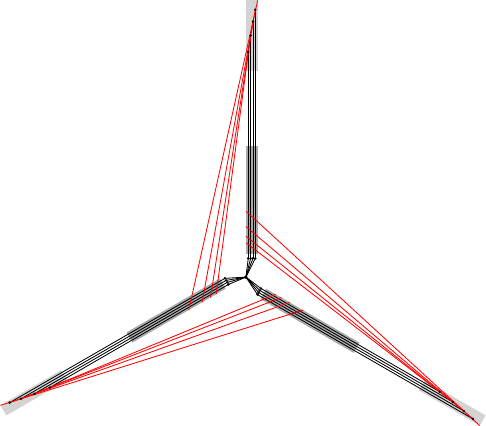}
      \label{fig:grr:crossing-nphard:constr:4}}
    \caption{Constructing~$\Gamma$ from~$\Pi$ for the proof of
      Theorem~\ref{thm:grr:crossing-nphard}.}
  \end{figure}

  We now describe the procedure to construct~$\Gamma$ from~$\Pi$ in
  detail. We will make sure that all vertices of~$\Gamma$ have
  rational coordinates with numerators and denominators in~$O(n^2)$.
  Let~$V_1 = \{ u_1, \dots, u_q\}$, $V_2 = \{ v_1, \dots, v_q\}$
  and~$V_3 = \{ w_1, \dots, w_q\}$. For the construction, we introduce
  dummy points~$u_0^\Gamma$, $u_{q+1}^\Gamma$, $v_0^\Gamma$,
  $v_{q+1}^\Gamma$, $w_0^\Gamma$, $w_{q+1}^\Gamma$, which do not lie
  on~$\Gamma$. For all~$i=0, \dots, q+1$, it will be
  $|o u_i^\Gamma| = |o v_i^\Gamma| = |o w_i^\Gamma|$. 

  We first show how to choose coordinates for
  points~$o, u_0^\Gamma, \dots, u_{q+1}^\Gamma$; see
  Fig.~\ref{fig:grr:crossing-nphard:constr:1}.  We approximate
  $120\dg$ rotation using the angle~$\alpha \approx 120.51\dg$
  with~$\cos \alpha = - \frac{33}{65}$
  and~$\sin \alpha = \frac{56}{65}$. The points $v_i^\Gamma$ are
  acquired from~$u_i^\Gamma$ by a clockwise rotation by~$\alpha$
  at~$o$, and the points $w_i^\Gamma$ are acquired from~$u_i^\Gamma$
  by a counterclockwise rotation by~$\alpha$ at~$o$. Then,
  $\angle u_i^\Gamma o v_i^\Gamma = \angle u_i^\Gamma o w_i^\Gamma =
  \alpha$
  and~$\angle v_i^\Gamma o w_i^\Gamma = 360\dg - 2 \alpha \approx
  118.98\dg$.

  Let point~$o$ have coordinates~$(0,0)$. For~$i=1, \dots, q$, let the
  first segment of the~$o$-$u_i^\Gamma$ path have its other endpoint
  in~$(i, c_1 q)$ for a constant~$c_1$. For~$i=0, \dots, i+1$,
  point~$u_i^\Gamma$ has $x$-coordinate~$i$. Let~$y_i$ denote
  the~$y$-coordinate of~$u_i^\Gamma$. We set~$y_0 = c_1 q + c_2 q^2$
  for a constant~$c_2$. For~$i = 1, \dots, q$, we
  set~$y_{i} = y_{i-1} + 2q + 1 - i$; see
  Fig.~\ref{fig:grr:crossing-nphard:constr:1}. Thus,
  for~$i=0, \dots, q+1$, points~$u_i^\Gamma$ lie on a parabola that
  opens down. Note that all vertices of~$\Gamma$ constructed so far
  are integers in~$O(n^2)$. We set~$c_1 = 5$ and~$c_2 = 40$.

  Next, we show how to construct the dents on the~$o$-$u_i^\Gamma$
  paths. For edge~$u_i v_j$ in~$G$, $i,j=1, \dots, q$, consider the
  straight line through~$v_{j-1}^\Gamma v_{j+1}^\Gamma$; see the
  dashed red line in Fig.~\ref{fig:grr:crossing-nphard:constr:2}
  for~$j=3$. Consider the intersection of this line and the vertical
  line through~$u_i^\Gamma$. The coordinates of that intersection are
  rational numbers with numerators and denominators in~$O(n^2)$. It is
  easy to show that this intersection has $y$-coordinates
  between~$\frac{c_2}{2} q = 20 q$
  and~$\frac{6}{5}(c_1 + c_2 q + \frac{3}{2}(q+1)) < 8 + 50 q$.

  At the intersection, we place a dent consisting of two segments; see
  Fig.~\ref{fig:grr:crossing-nphard:constr:3}. The first segment of
  the dent has positive slope and is orthogonal
  to~$v_{j-1}^\Gamma v_{j+1}^\Gamma$. Its projection on the~$x$ axis
  has length~$\frac{1}{2}$. The second segment has the negative slope
  of~$-5$. It is easy to verify that the line
  through~$v_{j}^\Gamma v_{j+2}^\Gamma$ (the upper red dashed line in
  Fig.~\ref{fig:grr:crossing-nphard:constr:3}) has distance at
  least~$\frac{c_2}{8} = 5$ from the lowest point of the
  dent. Therefore, the dent fits between the two dashed red
  lines. Note that all three vertices of the dent have coordinates
  that are rational numbers with numerators and denominators
  in~$O(n^2)$.

  By the choice of the slopes, no normal to either one of the dent
  segments crosses~$o w_k^\Gamma$ for~$k=0, \dots, q+1$. Furthermore,
  no normal on the second segment crosses $o v_k^\Gamma$
  for~$k=0, \dots, q+1$, and a normal to the first segment only
  crosses~$o v_k^\Gamma$ for~$k=j$. In this way, the dent ensures
  that~$u_i^\Gamma$ and~$v_j^\Gamma$ can not be in the same GRR, and
  it does not prohibit any other vertex pair~($u_k^\Gamma$
  and~$v_\ell^\Gamma$, $v_k^\Gamma$ and~$w_\ell^\Gamma$, $w_k^\Gamma$
  and~$v_\ell^\Gamma$, $k,\ell=1,\dots,q$) from being in the same
  GRR. Finally, for each leaf vertex~$u_i^\Gamma$, we add the missing
  segments on the vertical line through~$u_i^\Gamma$ to connect~$o$
  and~$u_i^\Gamma$ by a path. Analogously, we construct the
  $o$-$v_i^\Gamma$ and the $o$-$w_i^\Gamma$ paths.

  Note that by our construction, the dent normals do not cross other
  dents on the paths from~$o$ to the leaves from another partition;
  see Fig.~\ref{fig:grr:crossing-nphard:constr:4}, where the dents lie
  in the dark gray rectangles, and the crossings of dent normals and
  paths from~$o$ to the leaves from another partition lie in the light
  gray rectangles. It follows that for~$i,j=1,\dots,q$, the
  $o$-$u_i^\Gamma$ and the~$o$-$v_j^\Gamma$ path can be merged into
  one GRR, if no dent corresponding to edge~$u_i v_j$ in~$G$ exists on
  the $o$-$u_i^\Gamma$ path in~$\Gamma$.

  From the construction of~$\Gamma$, it follows that a pair of
  leaves~$x^\Gamma$ and~$y^\Gamma$ can be in the same GRR if and only
  if the corresponding vertices~$x,y$ are in different partitions
  of~$V$ and edge~$x y$ is not in~$G$. Therefore, triples of
  leaves~$x^\Gamma,y^\Gamma,z^\Gamma$ for which
  $x^\Gamma,y^\Gamma,z^\Gamma$ can be in the same GRR, are in one to
  one correspondence to independent triples
  from~$V_1 \times V_2 \times V_3$ in~$G$. Therefore, $\Gamma$ can be
  partitioned into~$q$ GRRs if and only if~$\Pi$ is a yes-instance of
  PIIT. Note that~$\Gamma$ can be constructed in polynomial time and
  that all coordinates of vertices in~$\Gamma$ are rational numbers
  with numerators and denominators in~$O(n^2)$.
\end{proof}

\subsection{Polynomial-time algorithms for restricted types of contacts}
\label{sec:trees:restricted-contacts}

We now make a restriction by only allowing non-crossing contacts.

First, assume~$T$ is split only at its vertices. As shown in
Section~\ref{sec:split}, we can drop this restriction and adapt our
algorithms to compute minimum or approximately minimum GRR
decompositions of plane straight-line tree drawings which allow
splitting tree edges at interior points. Note that the construction in
the proof of Lemma~\ref{lem:split} preserves the non-crossing property
of GRR contacts.

We start in Section~\ref{sec:trees:multicut} and use the well-known
problem \minmulticut to compute a 2-approximation for minimum \gtds
for the scenario in which GRRs are only allowed to have proper
contacts. A similar approach will be used in Section~\ref{sec:triang}
to compute minimum GRR decompositions of triangulated polygons.
After that, in Section~\ref{sec:dynprog}, we present an exact, but
more complex approach for computing \gtds, which also allows
non-crossing contacts.

\subsubsection{2-Approximation using Multicut}
\label{sec:trees:multicut}

\begin{figure}[tb]
  \hfill 
\subfloat[]{
  \includegraphics[page=1]{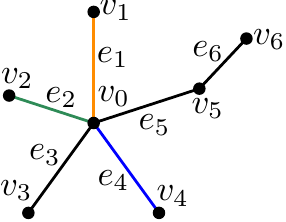}\label{fig:trees:minmult-1}}
  \hfill 
\subfloat[]{
  \includegraphics[page=2]{fig/minmult-trees.pdf}\label{fig:trees:minmult-2}}
\hfill\null
\caption{\protect\subref{fig:trees:minmult-1}~Tree drawing decomposed
  in GRRs. Edge pairs~$\{e_1, e_2\}$, \dots, $\{e_4, e_5\}$,
  $\{e_5, e_1\}$ as well as~$\{e_1, e_6\}$, $\{e_4, e_6\}$ are
  conflicting.  \protect\subref{fig:trees:minmult-2}~\minmulticut
  instance constructed according to the proof of
  Proposition~\ref{prop:trees:minmult}. No edge orientation respecting
  all paths between the terminals exists. Dashed edges form a
  solution.}
\end{figure}

We show how to partition the edges of~$T$ into a minimum number of \ic
components with proper contacts using \minmulticut on trees. Given an
edge-weighted graph~$G=(V,E)$ and a set of terminal
pairs~$\{(s_1,t_1)$, \dots,$(s_k,t_k)\}$, an edge set~$S \subseteq E$ is
a \emph{multicut} if removing~$S$ from~$G$ disconnects each
pair~$s_i,t_i$, $i=1, \dots,k$. A multicut is minimum if the total
weight of its edges is minimum.

For the complexity of \minmulticut on special graph types, see the
survey by Costa et al.~\cite{clr-mmmims-05}. Computing \minmulticut is
$\NP$-hard even for unweighted binary trees~\cite{gfr-mugd-2003}, but
has a polynomial-time
2-approximation for trees~\cite{gvy-pdaaifmt-1997}.

Consider a plane straight-line drawing of a tree~$T=(V,E)$. We
construct a tree~$T_M$ by subdividing every edge of~$T$ once as
follows. Tree~$T_M$ has a vertex~$n_v$ for each vertex~$v \in V$ and a
vertex~$n_e$ for each edge~$e \in E$. For each~$e=uv \in E$,
edges~$n_u n_e$ and~$n_e n_v$ are in~$T_M$.  The set $X$ of terminal
pairs contains a pair $(n_e,n_f)$ for each pair of conflicting edges
$e, f$ of~$T$. Let all edges of~$T_M$ have weight~1.

\begin{lemma}
  Let~$E'$ be a \minmulticut of~$T_M$ with respect to the terminal
  pairs~$X$ and let~$C_1^M, \dots, C_k^M$ denote the connected
  components of~$T_M - E'$. Then,
  components~$C_i = \{ e \in E \mid n_e \in C_i^M \}$ form a minimum
  GRR decomposition of~$T$.
\label{prop:trees:minmult}
\end{lemma}
\begin{proof}
  Consider a multicut~$E'$ of~$T_M$, $|E'| = k-1$. Consider a
  component~$C_i^M$. Then, the edges in~$C_i$ are conflict-free and
  form a connected subtree~$T_i$ of~$T$. Thus,~$T_i$ is a GRR by
  Lemma~\ref{lem:tree-grr}.

  Next, consider a GRR decomposition of~$T$ into~$k$ subtrees~$T_i =
  (V_i,E_i)$ with proper contacts. We create an edge set $S$ as follows. Assume~$T_i$, $T_j$ touch at
  vertex~$v \in V$. Let edge~$e = uv$ be in~$T_i$, and let~$v$ be a
  leaf in~$T_i$. Then we add edge~$n_e n_v$ of~$T_M$ to set~$S$; see
  Fig.~\ref{fig:trees:minmult-1} and~\ref{fig:trees:minmult-2}. It
  is~$|S|=k-1$. After removing~$S$ from~$T_M$, no connected component
  contains vertices~$n_{e_1}, n_{e_2}$ for a pair of conflicting
  edges~$e_1$, $e_2$. Thus,~$S$ is a multicut. 

  We have shown that GRR decompositions of~$T$ of size~$k$ are in
  one-to-one correspondence with the multicuts of~$T_M$ of size $k-1$.
  Therefore, minimum multicuts correspond to minimum GRR
  decompositions, and it follows that~$C_i$ form a minimum GRR
  decomposition of~$T$.
\end{proof}
Note that \minmulticut can be solved in polynomial time in directed
trees~\cite{clr-gamimrt-03}, i.e., trees whose edges can be directed
such that for each terminal pair~$(s_i,t_i)$, the~$s_i$-$t_i$ path is
directed. We note that this result cannot be applied in our context,
since we can get \minmulticut instances for which no such orientation
is possible, see Fig.~\ref{fig:trees:minmult-2}. However, using the
approximation algorithm from~\cite{gvy-pdaaifmt-1997}, we obtain the
following result.

\begin{corollary}
  Given a plane straight-line drawing of a tree~$T=(V,E)$, a partition
  of~$E$ into $2\cdot \textnormal{OPT} -1$ \ic subtrees of~$T$ having
  only proper contacts can be computed in time polynomial in~$n$,
  where $\textnormal{OPT}$ is the minimum size of such a partition.
\end{corollary}

\subsubsection{Optimal solution}
\label{sec:dynprog}

\newcommand{\ThmNonCrossingText}{Given a plane straight-line
  drawing of a tree~$T=(V,E)$, a partition of~$E$ into a minimum number
  of \ic subtrees of~$T$ (minimum \gtd) having only \emph{non-crossing}
  contacts can be computed in time $O(n^6)$.}

\newcommand{\ThmProperText}{Given a plane straight-line drawing of a
  tree~$T=(V,E)$, a partition of~$E$ into a minimum number of \ic
  subtrees of~$T$ (minimum \gtd) having only \emph{proper} contacts can be
  computed in time $O(n^6)$.}

In the following we show how to find a minimum GRR partition with only
non-crossing contacts in polynomial time. As is the case with minimum
partitions of simple hole-free polygons into convex~\cite{cd-ocd-85}
or star-shaped~\cite{k-dpsc-85} components, our algorithm is based on
dynamic programming. We describe the dynamic program in detail and use
it to find minimum \gtds for the setting as in
Section~\ref{sec:trees:multicut}, as well as for the setting in which
non-proper, but non-crossing contacts of GRRs are allowed. First, we
shall prove the following theorem.

\begin{theorem}
\label{thm:trees:partition-edges-non-crossing}
\ThmNonCrossingText
\end{theorem}

At the end of Section~\ref{sec:dynprog}, we modify our dynamic program
slightly to prove Theorem~\ref{thm:trees:partition-edges}, which shows
the same result for the setting in which only partitions with proper
contacts are considered.

\begin{theorem}  
  \label{thm:trees:partition-edges}
\ThmProperText
\end{theorem}

%
Let~$T$ be rooted.  For each vertex~$u$ with parent~$\parent{u}$,
let~$T_u$ be the subtree of~$u$ together with edge~$\parent{u} u$. We
shall use the following definition.

\begin{definition}[root component]
  Given a GRR partition of the edges of a rooted tree~$T'$, we call
  all GRRs containing the root of~$T'$ the \emph{root components}. If
  the root of~$T'$ has degree~1, every GRR partition of~$T'$ has one
  unique root component.
  \label{def:root-component}
\end{definition}
A minimum partition is constructed from the solutions of subinstances
as follows.
Let~$u_1, \dots, u_d$ be the children of~$u$. For subtrees~$T_{u_1}$,
\dots, $T_{u_d}$ whose only common vertex is~$u$, a minimum partition
$P'$ of~$T' = \bigcup_i T_{u_i}$ induces
partitions~$P_{i}$ of~$T_{u_i}$. Furthermore,
$P'$ is created by choosing~$P_{i}$ as
partitions of~$T_{u_i}$ and possibly merging some of the root
components of~$T_{u_i}$, $i=1, \dots, d$.
Note that~$P_{i}$ is not necessarily a minimum partition of~$T_{u_i}$,
if~$P_{i}$ allows us to merge more root components than a minimum
partition of~$T_{u_i}$ would allow. Therefore, for every~$u$ we shall
store minimum partitions of~$T_u$ for various possibilities of the
root component of~$T_u$. For the sake of uniformity, we choose a
vertex with degree~1 as the root of~$T$.

Given a tree root, the number of different subtrees it could be
contained in may be exponential, e.g., it is~$\Theta(2^n)$ in a star.
The key observation for our algorithm is that we do not need to store
a partition for each possible root component. We require the following
notation.

\begin{definition}[Path clockwise between]
  Consider directed non-crossing paths~$\rho_1$, $\rho_2$, $\rho_3$
  with common origin~$r$, endpoints $t_1$, $t_2$, $t_3$ and, possibly,
  common prefixes. Let~$V_i$ be vertices of~$\rho_i$, $i =1,2,3$, and
  let~$T$ be the tree formed by the union of~$\rho_1,\rho_2$
  and~$\rho_3$. We say that~$\rho_2$ is \emph{clockwise
    between}~$\rho_1$ and~$\rho_3$, if the clockwise traversal of the
  outer face of~$T$ visits~$t_1$, $t_2$, $t_3$ in this order; see
  Fig.~\ref{fig:path-clockwise-between}.
\label{def:path-clockwise-between}
\end{definition}

Note that in Definition~\ref{def:path-clockwise-between} the three paths may
(partially) coincide.
Lemma~\ref{lem:trees:combine:paths} shows that to decide whether a
union of two subtrees is \ic, it is sufficient to consider only the
two pairs of ``outermost'' root-leaf paths of each subtree. This
result is crucial for limiting the number of representative
decompositions that need to be considered during our dynamic
programming approach.
The statement of the lemma is illustrated in
Fig.~\ref{fig:trees:combine:paths}.

\begin{figure}[tb]
  \hfill
  \subfloat{\includegraphics[scale=0.9,page=3]{fig/lemma-trees-dynprog.pdf} \label{fig:path-clockwise-between}}
  \hfill
  \subfloat[]{\includegraphics[scale=0.8,page=1]{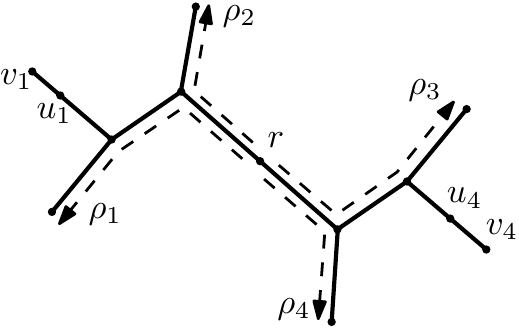}\label{fig:trees:combine:paths}}
  \hfill\null
  \caption{\protect\subref{fig:path-clockwise-between}~Path~$\rho_2$
    is clockwise between paths~$\rho_1$ and~$\rho_3$.
    \protect\subref{fig:trees:combine:paths}~Statement of
    Lemma~\ref{lem:trees:combine:paths}.}
\end{figure}

\newcommand{\lemTreesCombinePathsText}{%
  Let~$T_1$, $T_2$ be \ic trees sharing a single
  vertex~$r$. Let all tree edges be directed away from~$r$. Let
  paths~$\rho_1$, $\rho_2$ in~$T_1$ and~$\rho_3$, $\rho_4$ in~$T_2$ be
  paths from~$r$ to a leaf, such that:
  \begin{compactenum}
  \item[-] every directed path from~$r$ in~$T_1$ is clockwise
    between~$\rho_1$ and~$\rho_2$;
  \item[-] every directed path from~$r$ in~$T_2$ is clockwise
    between~$\rho_3$ and~$\rho_4$;
  \item[-] for~$i=1,\dots,4$, path~$\rho_i$ is clockwise
    between~$\rho_{i-1}$ and~$\rho_{i+1}$ (indices modulo 4).
  \end{compactenum}

  Then, $\rho_1 \cup \rho_2 \cup \rho_3 \cup \rho_4$ is \ic if and
  only if $T_1 \cup T_2$ is \ic.  }
\begin{lemma}
  \label{lem:trees:combine:paths}
\lemTreesCombinePathsText
\end{lemma}

\begin{proof}
  Consider trees~$T_1$, $T_2$ and paths~$\rho_1,\dots,\rho_4$
  satisfying the condition of the lemma; see
  Fig.~\ref{fig:trees:combine:paths} for a sketch. Note that~$\rho_1$
  and~$\rho_2$ may have common prefixes, and so may~$\rho_3$
  and~$\rho_4$.
  Assume the four paths $\rho_1,\dots,\rho_4$ are drawn with
  increasing chords, but the union~$T'$ of the trees~$T_1$ and~$T_2$
  is not. Then, there exist edges~$u_1 v_1$ in~$T_1$ and~$u_4 v_4$
  in~$T_2$, such that the normal~$\ell$ to~$u_1 v_1$ at~$u_1$ crosses
  edge~$u_4 v_4$.
  \begin{figure}[htb]
    \hfill
    \subfloat[]{\includegraphics[page=2]{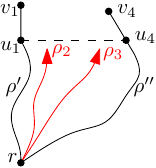}\label{fig:trees:combine:paths:proof:1}}
    \hfill
    \subfloat[]{\includegraphics[page=3]{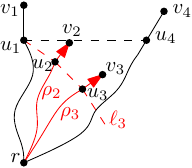}\label{fig:trees:combine:paths:proof:2}}\hfill\null
    \caption{Constructions in Lemma~\ref{lem:trees:combine:paths}.}\label{fig:trees:combine:paths:proof}
  \end{figure}
  \begin{claim}
    \Wlg, we may assume the following; see
    Fig.~\ref{fig:trees:combine:paths:proof}.
    \begin{inparaenum}[(i)]
    \item Edge~$u_1 v_1$ points vertically upwards,
    \item edge~$u_4 v_4$ is the first edge on the $r$-$v_4$
      path~$\rho''$ crossed by~$\ell$ and points upwards,
    \item vertex~$u_4$ is on~$\ell$ and to the right of~$u_1$.
    \end{inparaenum}
  \end{claim} 

  We ensure~(i) by rotation. Then, point~$r$ is below~$\ell$ (or on
  it), since the $r$-$v_1$ path~$\rho'$ is \ic. For~(ii), we
  choose~$u_4 v_4$ as the first edge with this property. If it points
  downward, there is an edge on the~$r$-$u_4$ path crossed
  by~$\ell$. For~(iii), if~$\ell$ crosses~$u_4 v_4$ in an interior
  point~$p$, we subdivide the edge at~$p$ and replace~$u_4 v_4$ by~$p
  v_4$. If~$u_4$ is left of~$u_1$, we mirror the drawing
  horizontally. This proves the claim.

  First, assume that~$v_1$, $v_4$ are not on paths $\rho_1, \dots,
  \rho_4$. Recall that two of the paths~$\rho_1, \dots, \rho_4$ (\wlg,
  $\rho_2$ and~$\rho_3$) are between~$\rho'$ and~$\rho''$. Let $u_2
  v_2$ and~$u_3 v_3$ be the last two edges on~$\rho_2$ and~$\rho_3$,
  respectively. Note that $\ray{u_1}{v_1}$ and~$\ray{u_2}{v_2}$ must
  diverge, and so must~$\ray{u_2}{v_2}$ and~$\ray{u_3}{v_3}$. If $u_4
  v_4$ points upwards and to the left as in
  Fig.~\ref{fig:trees:combine:paths:proof:1}, then~$\ray{u_3}{v_3}$
  and~$\ray{u_4}{v_4}$ must converge; a contradiction. Thus, $u_2
  v_2$, $u_3 v_3$ and~$u_4 v_4$ point upwards and to the right; see
  Fig.~\ref{fig:trees:combine:paths:proof:2}.
  Since~$T_1$ as well as the union of~$\rho_1$ and~$\rho_2$ is \ic,
  the angles~$\angle v_1 u_1 u_2$, $\angle u_1 u_2 v_2$, $\angle v_2
  u_2 u_3$ and~$\angle u_2 u_3 v_3$ are between~$90\dg$
  and~$180\dg$. Therefore, vertices~$u_2$ and~$u_3$ must lie
  below~$\ell$. Let~$\ell_3$ be the normal to~$u_3 v_3$
  at~$u_3$. Since~$T_2$ is drawn with increasing chords, $u_4 v_4$
  must lie below~$\ell_3$, a contradiction.

  The proof works similarly if~$u_1 v_1$ is on~$\rho_2$ (by
  identifying~$u_1 v_1$ and~$u_2 v_2$), and the remaining cases are
  symmetric.
\end{proof}

We now describe our dynamic programs for proper and non-crossing contacts in detail. We first give an overview of the general approach, then describe the non-crossing case and afterwards modify it for proper contacts.
For a root component~$R$ of~$T_u$, let the \emph{leftmost path} (or,
respectively, the \emph{rightmost path}) be the simple path in~$R$
starting at~$\parent{u}$ which always chooses the next
counterclockwise (clockwise) edge.

The basic idea of the dynamic program is as follows. For a given
subtree~$T_u$, we store the sizes of the minimum \gtds
of~$T_u$ for different possibilities of the root component. We combine
these solutions to compute minimum \gtds of bigger
subtrees. For this step, we must be able to test which root components
can be merged into one GRR. Instead of storing the partition sizes for
\emph{all} possible root components, we only store the minimum
partition size for each combination of the leftmost and rightmost path
of the root component. Thus, for each~$T_u$, we only store~$O(n^2)$
partition sizes.  Note that this is sufficient, since by
Lemma~\ref{lem:trees:combine:paths} the question whether two root
components can be merged depends only on their leftmost and rightmost
paths.

If~$u$ is the root of a subtree~$T'$ and has degree~2 or greater
in~$T'$, there might be several root components in a partition
of~$T'$, i.e., GRRs containing~$u$. Let~$R$ be some fixed root
component of the considered \gtd. If~$u$ has degree~2 or greater
in~$R$, then we need a reference direction to define the leftmost and
rightmost paths of~$R$. Let~$\rho_l$ be the leftmost path of the
rooted tree~$R+\parent{u}u$. Note that~$\rho_l$ contains the
edge~$\parent{u}u$. Then, the leftmost path of~$R$
is~$\rho_l - \parent{u}u$. The rightmost path of~$R$ is defined
analogously.

Recall that $T_u$ is the subtree of~$u$ together with
edge~$\parent{u} u$. For each pair of vertices~$t_i, t_j$ in~$T_u$,
cell~$\tau[u, t_i, t_j]$ of a table~$\tau$ stores the size of a
minimum GRR decomposition of~$T_u$, in which the root component has
the $\parent{u}$-$t_i$ path and the $\parent{u}$-$t_j$ path as its
leftmost and rightmost path, respectively. Cell~$\tau[u]$ stores the
size of a minimum GRR decomposition of~$T_u$. It
is~$\tau[u] = \min_{t_i,t_j} \tau[u, t_i, t_j]$. For simplicity, we
set~$\min \emptyset = \infty$.

Clearly, for each leaf~$u$, $\tau[u,u,u]=1$,
and~$\tau[u,t_i,t_j] = \infty$ for all other values
of~$t_i,t_j$. Let~$v$ be the only neighbor of the root~$r$ of the
tree~$T$. Then,~$\tau[v]$ is the size of a minimum GRR decomposition
of~$T$. We show how to compute~$\tau$ bottom-up.

For ease of presentation, we use the following notation. Vertex~$u$
is not a leaf and has children~$u_1, \dots, u_d$. Let~$\parent{u}$,
$u_1$, \dots, $u_d$ have this clockwise order around~$u$.
Let~$t_i \neq u$ be a vertex in~$T_{u_i}$. We
define~$t_j, t_k, t_\ell$ analogously for
$1 \leq i \leq j \leq k \leq \ell \leq d$. Let~$\rho_i$ be the
$u$-$t_i$ path.

We consider two settings: allowing arbitrary non-crossing contacts and
allowing only proper contacts. The dynamic programs for the two cases
are very similar, and the program for arbitrary non-crossing contacts
is slightly more complex. To reduce duplication, we first present the
program for arbitrary non-crossing contacts, and later show how to
modify it for the case when only proper contacts are allowed.

\subsubsection{Non-crossing contacts}

Recall that vertex~$u$ can live in a root component~$R$ together with
non-consecutive children~$u_i$, $u_\ell$, $i < \ell$.  If arbitrary
non-crossing contacts are allowed, some nodes from
$u_{i+1}, \dots, u_{\ell-1}$ that are not in~$R$ can also be in one
GRR. Therefore, after choosing the root component~$R$ of~$T_u$, we
must be able to recursively compute the minimum size of a partition of
the union of~$T_{u_j}$, $u_j \notin R$. We introduce additional tables
for this purpose.

In addition to the table~$\tau$ storing the
values~$\tau[u, t_i, t_j]$, we use tables~$\sigma_\Delta$
for~$\Delta=1, \dots, 4$, as well as tables~$\sigma$ and~$\sigma_M$.
These additional tables will be used to formulate the recurrences
for~$\tau$. For fixed~$u$, $i$, $j$, the corresponding values of
$\sigma_\Delta$, $\sigma$ and~$\sigma_M$ denote the sizes of minimum
\gtds of~$T_{u_i} \cup T_{u_{i+1}} \cup \dots \cup T_{u_j}$
with certain properties. Table~$\sigma_\Delta$ considers different
possibilities of the leftmost and rightmost paths of the root
components as well as the degree~$\Delta$ of~$u$ in the root
component.
Recall that in an \ic tree drawing, every vertex has degree at most~4.
Formally, the value~$\sigma_\Delta[u,t_i,t_j]$ denotes the minimum
number of GRRs in a \gtd of the
tree~$T_{u_i} \cup T_{u_{i+1}} \cup \dots \cup T_{u_j}$, in which
there exists a GRR~$R$ with the rightmost path~$u$-$t_i$ and leftmost
path~$u$-$t_j$ and in which~$u$ has degree~$\Delta$ in~$R$.

For some recurrences, we need to aggregate the various possibilities
stored in~$\sigma_\Delta$. For this purpose, we use tables~$\sigma$
and~$\sigma_M$ as follows. The value~$\sigma$ is the minimum
of~$\sigma_\Delta$ over all values of~$\Delta$. We
define~$\sigma[u,t_i,t_j]$ as
$\sigma[u,t_i,t_j] = \min_{\Delta = 1, \dots, 4}
\sigma_\Delta[u,t_i,t_j]$.

The value~$\sigma_M$ stores the minimum over all combinations of the
leftmost and rightmost paths. Thus, it stores the size of the minimum
partition of~$T_{u_i} \cup \dots \cup T_{u_j}$, regardless of the root
component.
Formally,~$\sigma_M[u,i,j]$ denotes the minimum number of GRRs in a
\gtd of~$T_{u_i} \cup \dots \cup T_{u_j}$.
Note that the arguments of~$\sigma_M[u, \cdot, \cdot]$ are
indices~$i$,~$j$ of a pair of children of~$u$, and the arguments
of~$\sigma_\Delta[u, \cdot, \cdot]$ and~$\sigma[u, \cdot, \cdot]$ are
a pair of vertices in~$T_{u_i} \cup \dots \cup
T_{u_j}$.

In the following recurrences, for a fixed pair of vertices~$t_i$
and~$t_\ell$, all possibilities for~$t_j$ and~$t_k$ are considered,
such that both paths~$\rho_j$ and~$\rho_k$ are clockwise
between~$\rho_i$ and~$\rho_\ell$. We test whether root
components~$R_1$ with the leftmost and rightmost paths~$\rho_i$
and~$\rho_j$ and~$R_2$ with the leftmost and rightmost paths~$\rho_k$
and~$\rho_\ell$ can be merged to a single GRR. We show that this
covers all representative possibilities for a root component of a \gtd
of~$T_{u_i} \cup \dots \cup T_{u_\ell}$ to have the leftmost and
rightmost paths~$\rho_i$ and~$\rho_\ell$, respectively.

\newcounter{rec}
\newcommand{\nextrec}[1]{\refstepcounter{rec}(\arabic{rec})\label{#1}}

\begin{lemma}
\label{lem:noncrossing:rec}
  We have the recurrences
  \begin{compactenum}
  \item[\nextrec{noncrossing:rec:sigma1}]  
    $\sigma_1[u, t_i, t_j] = \sigma[u, t_i, t_j] = \tau[u_i, t_i, t_j]$
    for all~$t_i, t_j \neq u$ in~$T_{u_i}$, $i=1, \dots, d$;
  \item[\nextrec{noncrossing:rec:sigma11}]
    $\sigma_M[u, i, i] = \tau[u_i]$ for all~$i=1, \dots, d$;
  \item[\nextrec{noncrossing:rec:sigma2}]
    $\sigma_2[u,t_i,t_\ell] = \min_{t_j, t_k} \{ \sigma_1[u,t_i,t_j] + \sigma_M[u,
    j+1, k-1] + \sigma_1[u, t_k, t_\ell] - 1 \}$;
  \item[\nextrec{noncrossing:rec:sigma3}]
    $\sigma_3[u,t_i,t_\ell] = \min \{ \min_{t_j,t_k} \{
    \sigma_2[u,t_i,t_j] + \sigma_M[u, j+1, k-1] + \sigma_1[u, t_k,
    t_\ell] - 1 \},$

    \hspace{29.9mm}$\min_{t_j,t_k} \{ \sigma_1[u,t_i,t_j] +
    \sigma_M[u, j+1, k-1] + \sigma_2[u, t_k, t_\ell] - 1 \} \}$;

  \item[\nextrec{noncrossing:rec:sigma4}]
    $\sigma_4[u,t_i,t_\ell] = \min_{t_j, t_k} \{ \sigma_1[u, t_i, t_i] +
    \sigma_M[u, i+1, j-1] + \sigma_1[u, t_j, t_j]$

    \hspace{18.9mm}$+\sigma_M[u, j+1, k-1] + \sigma_1[u, t_k, t_k]
    +\sigma_M[u, k+1, \ell-1] + \sigma_1[u, t_\ell, t_\ell] \}-3$.
  \end{compactenum}
  The minimizations in lines~\eqref{noncrossing:rec:sigma3}
  and~\eqref{noncrossing:rec:sigma4} only consider
  vertices~$t_j, t_k$, such that
  $\rho_i \cup \rho_j \cup \rho_k \cup \rho_\ell$ is \ic.
\end{lemma}

\begin{figure}[tb]
  \hfill
  \subfloat[]{\includegraphics[page=1, scale=0.80]{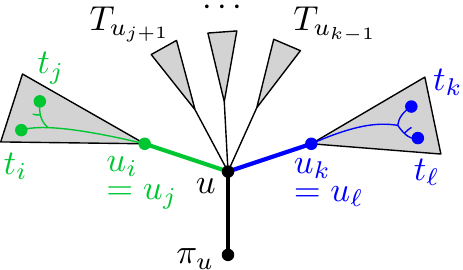} \label{fig:noncrossing:rec:1}}\hfill
  \subfloat[]{\includegraphics[page=2, scale=0.80]{fig/lemma-trees-dynprog-noncrossing} \label{fig:noncrossing:rec:2}}\hfill
  \subfloat[]{\includegraphics[page=3, scale=0.80]{fig/lemma-trees-dynprog-noncrossing} \label{fig:noncrossing:rec:3}}\hfill\null
  \caption{Recurrences in
    Lemma~\ref{lem:noncrossing:rec}. \protect\subref{fig:noncrossing:rec:1}~recurrence~\eqref{noncrossing:rec:sigma1};
    \protect\subref{fig:noncrossing:rec:2}~recurrence~\eqref{noncrossing:rec:sigma2}
    for the case~$m=j$;
    \protect\subref{fig:noncrossing:rec:2}~recurrence~\eqref{noncrossing:rec:sigma2}
    for the case~$m=k$.}
  \label{fig:noncrossing:rec}
\end{figure}

\begin{proof}
  Consider recurrence~\eqref{noncrossing:rec:sigma1} and a \gtd of
  $T_{u_i} \cup \dots \cup T_{u_j}$ of size~$x$ with root
  component~$R$, such that~$R$ has~$u$-$t_i$ and~$u$-$t_j$ as its
  leftmost and rightmost paths, respectively. Since~$u$ has degree~$1$
  in~$R$, it must be~$i=j$. Thus, this partition is a \gtd
  of~$T_{u_i}$ with~$R$ as the root component, so by definition
  of~$\tau$ we have $\tau[u, t_i, t_j] \leq x$. Thus, we have
  $\sigma_1[u, t_i, t_j] \geq \tau[u_i, t_i, t_j]$. Conversely,
  consider a \gtd of~$T_{u_i}$, such that its root component~$R$
  has~$u$-$t_i$ and~$u$-$t_j$ as its leftmost and rightmost
  paths. Thus, $t_i$ and~$t_j$ are both in~$T_{u_i}$, and vertex~$u$
  has degree~$1$ in~$R$. By the definition of~$\sigma_1$, this
  partition has size at least~$\sigma_1[u, t_i, t_j]$. Thus, we have
  $\sigma_1[u, t_i, t_j] \leq \tau[u_i, t_i, t_j]$. Finally, since
  for~$i=j$ we have $T_{u_i} \cup \dots \cup T_{u_j} = T_{u_i}$,
  vertex~$u$ can only have degree~$1$ in the root component of a \gtd,
  so we have $\sigma_1[u, t_i, t_j] = \sigma[u, t_i, t_j]$. Thus,
  recurrence~\eqref{noncrossing:rec:sigma1} holds.

  Recurrence~\eqref{noncrossing:rec:sigma11} holds trivially, since by
  the definitions of~$\sigma_M$ and~$\tau[ \cdot ]$, both
  $\sigma_M[u, i, i]$ and~$\tau[u_i]$ denote the size of the minimum
  GRR partition of~$T_{u_i}$.

  Consider recurrence~\eqref{noncrossing:rec:sigma2} and a \gtd~$P$ of
  $T_{u_i} \cup \dots \cup T_{u_\ell}$ of size~$x$ with root
  component~$R$. Again, let~$R$ have~$u$-$t_i$ and~$u$-$t_\ell$ as its
  leftmost and rightmost paths, respectively. Let~$u$ have degree~$2$
  in~$R$. Therefore,~$i \neq \ell$, and~$R$ only consists of two
  parts~$R_1, R_2$ (green and blue in
  Fig.~\ref{fig:noncrossing:rec:1}, respectively), such that~$R_1$ is
  contained in~$T_{u_i}$ and~$R_2$ is contained
  in~$T_{u_\ell}$. Partition~$P$ induces a \gtd~$P_1$ of~$T_{u_i}$ of
  size~$x_1$, a \gtd~$P_2$ of~$T_{u_\ell}$ of size~$x_2$ and a
  \gtd~$P_3$ of~$T_{u_{i+1}} \cup \dots \cup T_{u_{\ell-1}}$ of
  size~$x_3$. Since $R_1 \cup R_2 = R$, we have
  $x = x_1 + x_2 + x_3 -1$. Let~$u_j$ be a vertex in~$R_1$, such
  that~$u$-$u_j$ is the rightmost path of~$R_1$. Let~$u_k$ be the
  vertex in~$R_2$, such that~$u$-$u_k$ is the leftmost path
  of~$R_2$. The subtree
  $\rho_i \cup \rho_j \cup \rho_k \cup \rho_\ell$ is contained in~$R$
  and, therefore, is \ic. By the definition of~$\sigma_1$
  and~$\sigma_M$, we have $\sigma_1[u,t_i,t_j] \leq x_1$,
  $\sigma_1[u, t_k, t_\ell] \leq x_2$ and
  $\sigma_M[u, j+1, k-1] \leq x_3$. Thus, the right part of
  recurrence~\eqref{noncrossing:rec:sigma2} is at most~$x$, so the
  right side is upper bounded by the left side.

  Conversely, let the right side of
  recurrence~\eqref{noncrossing:rec:sigma2} be less than~$\infty$.
  Let~$j,k$, $t_j, t_k$ be chosen such that the minimum on the right
  side is realized. Then,
  $\rho_i \cup \rho_j \cup \rho_k \cup \rho_\ell$ is
  \ic. Let~$\sigma_1[u,t_i,t_j] = x_1$, and let~$P_1$ be a \gtd of
  size~$x_1$ realizing the minimum in the definition
  of~$\sigma_1[u,t_i,t_j]$. Let~$R_1$ be the root component
  of~$P_1$. Then, $R_1$ has leftmost and rightmost paths~$u$-$t_i$
  and~$u$-$t_j$ respectively.
  Analogously, let~$\sigma_1[u,t_k,t_\ell] = x_2$, and let~$P_2$ be a
  \gtd of size~$x_2$ realizing the minimum in the definition
  of~$\sigma_1[u,t_k,t_\ell]$. Let~$R_2$ be the root component
  of~$P_2$. Then, $R_2$ has leftmost and rightmost paths~$u$-$t_k$
  and~$u$-$t_\ell$ respectively. Finally, let~$P_3$ be a \gtd of
  size~$x_3$ realizing the minimum in the definition
  of~$\sigma_M[u, j+1, k-1]$.
  By Lemma~\ref{lem:trees:combine:paths}, $R_1 \cup R_2$ is
  \ic. Consider the \gtd~$P$ formed by taking the union
  of~$P_1$, $P_2$ and~$P_3$ and merging~$R_1$ and~$R_2$. Partition~$P$
  has size~$x_1 + x_2 + x_3 - 1$. Its root component~$R$ has leftmost
  and rightmost paths~$u$-$t_i$ and~$u$-$t_\ell$ respectively, and~$u$
  has degree~2 in~$R$. Thus, by the definition
  of~$\sigma_2[u,t_i,t_\ell]$, it
  is~$\sigma_2[u,t_i,t_\ell] \leq x_1 + x_2 + x_3 - 1$. Thus, the left
  side of recurrence~\eqref{noncrossing:rec:sigma2} is upper bounded
  by its right side. Therefore,
  recurrence~\eqref{noncrossing:rec:sigma2} holds.

  Next, consider recurrence~\eqref{noncrossing:rec:sigma3} and a GRR
  partition~$P$ of $T_{u_i} \cup \dots \cup T_{u_\ell}$ of size~$x$
  with root component~$R$. Once again, let~$R$ have~$u$-$t_i$
  and~$u$-$t_\ell$ as its leftmost and rightmost paths,
  respectively. Let~$u$ have degree~$3$ in~$R$. Therefore, it
  is~$i \neq \ell$. 
  In addition to~$u_i$ and~$u_\ell$, the GRR~$R$ must contain another
  child~$u_m$ of~$u$, such that~$i < m < \ell$. We can partition~$R$
  into two GRRs~$R_1$ and~$R_2$, such that~$u_i$ is in~$R_1$, $u_\ell$
  in~$R_2$ and~$u_m$ is either in~$R_1$ or in~$R_2$. First, assume
  $u_m$ is in~$R_1$; see Fig.~\ref{fig:noncrossing:rec:2}. The other
  case is symmetric; see Fig.~\ref{fig:noncrossing:rec:3}. We
  choose~$j=m$. Let~$t_j$ be a vertex in~$T_{u_j}$, such that
  $u$-$t_j$ is the rightmost path of~$R_1$. Let~$t_k$ be a vertex
  in~$T_{u_\ell}$, such that~$u$-$t_k$ is the leftmost path
  in~$R_2$. Note that in this case, $t_k$ and~$t_\ell$ are in the same
  subtree~$T_{u_k} = T_{u_\ell}$.
  We can split the partition~$P$ into GRR partitions~$P_1$
  of~$T_{u_i} \cup \dots \cup T_{u_j}$ of size~$x_1$, $P_2$
  of~$T_{u_\ell}$ of size~$x_2$ and~$P_3$
  of~$T_{u_{j+1}} \cup \dots \cup T_{u_{k-1}}$ of size~$x_3$.  It
  holds: $R = R_1 \cup R_2$, and apart from~$R$, no other GRR in~$P$
  is split, since the contacts are non-crossing. Thus, it
  is~$x = x_1 + x_2 + x_3 - 1$. By definition,
  $\sigma_2[u,t_i,t_j] \leq x_1$, $\sigma_1[u, t_k, t_\ell] \leq x_2$
  and~$\sigma_M[u, j+1, k-1] \leq x_3$. Therefore, the right side of
  recurrence~\eqref{noncrossing:rec:sigma3} is at most~$x$. The same
  holds for the symmetric case in which $u_m$ is in~$R_2$ by analogous
  arguments. Thus, the right side of
  recurrence~\eqref{noncrossing:rec:sigma3} is upper bounded by its
  left side.

  Conversely, let the right side of
  recurrence~\eqref{noncrossing:rec:sigma3} be less than~$\infty$.
  Let~$j,k$, $t_j, t_k$ be chosen such that the minimum on the right
  side is realized. First, assume it is realized by
  $\sigma_2[u,t_i,t_j] + \sigma_M[u, j+1, k-1] + \sigma_1[u, t_k,
  t_\ell] - 1$.  Then, $\rho_i \cup \rho_j \cup \rho_k \cup \rho_\ell$
  is \ic.
  Let~$\sigma_2[u,t_i,t_j] = x_1$, and let~$P_1$ be a GRR partition of
  size~$x_1$ realizing the minimum in the definition
  of~$\sigma_2[u,t_i,t_j]$. Let~$R_1$ be the root component
  of~$P_1$. Then, $R_1$ has leftmost and rightmost paths~$u$-$t_i$
  and~$u$-$t_j$ respectively. The degree of~$u$ in~$R_1$ is~$2$, and
  the vertices~$t_i$ and~$t_j$ must lie in different subtrees
  $T_{u_i}$ and~$T_{u_j}$, respectively.
  Analogously, let~$\sigma_1[u,t_k,t_\ell] = x_2$, and let~$P_2$ be a
  GRR partition of size~$x_2$ realizing the minimum in the definition
  of~$\sigma_1[u,t_k,t_\ell]$. Let~$R_2$ be the root component
  of~$P_2$. Then, $R_2$ has leftmost and rightmost paths~$u$-$t_k$
  and~$u$-$t_\ell$ respectively. Finally, let~$P_3$ be a GRR partition
  of size~$x_3$ realizing the minimum in the definition
  of~$\sigma_M[u, j+1, k-1]$.
  By Lemma~\ref{lem:trees:combine:paths}, $R_1 \cup R_2$ is
  \ic. Consider the GRR partition~$P$ formed by taking the union
  of~$P_1$, $P_2$ and~$P_3$ and merging~$R_1$ and~$R_2$. Partition~$P$
  has size~$x_1 + x_2 + x_3 - 1$. Its root component~$R$ has leftmost
  and rightmost paths~$u$-$t_i$ and~$u$-$t_\ell$, respectively, and~$u$
  has degree~3 in~$R$. Therefore, by the definition
  of~$\sigma_3[u,t_i,t_\ell]$, it
  is~$\sigma_3[u,t_i,t_\ell] \leq x_1 + x_2 + x_3 - 1$. Thus, the left
  side of recurrence~\eqref{noncrossing:rec:sigma2} is upper bounded
  by its right side.
  The same holds for the symmetric case in which the minimum on the
  right side is realized by
  $\sigma_1[u,t_i,t_j] + \sigma_M[u, j+1, k-1] + \sigma_2[u, t_k,
  t_\ell] - 1$.
  Therefore, recurrence~\eqref{noncrossing:rec:sigma3} holds.

  Finally, consider recurrence~\eqref{noncrossing:rec:sigma4} and a
  \gtd~$P$ of $T_{u_i} \cup \dots \cup T_{u_\ell}$ of size~$x$ with
  root component~$R$. Once again, let~$R$ have~$u$-$t_i$
  and~$u$-$t_\ell$ as its leftmost and rightmost paths,
  respectively. Let~$u$ have degree~$4$ in~$R$. Then,~$R$ is a
  subdivision of~$K_{1,4}$~\cite{acglp-sag-12}. Let~$t_j$ and~$t_k$ be
  the other two leaves of~$R$ lying in the subtrees $T_{u_j}$
  and~$T_{u_k}$ respectively, for $1 \leq i < j < k < \ell \leq d$.
  Then, we can split~$P$ into~$7$ \gtds~$P_1$, \dots, $P_7$ as
  follows. Partitions~$P_1$, $P_2$, $P_3$, $P_4$ are \gtds of
  subtrees~$T_{u_i}$, $T_{u_j}$, $T_{u_k}$ and~$T_{u_\ell}$,
  respectively, with the respective sizes~$x_1$, $x_2$, $x_3$, $x_4$
  and paths~$u$-$u_i$, $u$-$u_j$, $u$-$u_k$ and~$u$-$u_\ell$ as the
  respective root components. Partitions~$P_5$, $P_6$, $P_7$ are \gtds
  of $T_{u_{i+1}} \cup \dots \cup T_{u_{j-1}}$,
  $T_{u_{j+1}} \cup \dots \cup T_{u_{k-1}}$ and
  $T_{u_{k+1}} \cup \dots \cup T_{u_{\ell-1}}$, respectively, with
  respective sizes~$x_5$, $x_6$ and~$x_7$. The root component~$R$ is
  split into the four paths~$u$-$u_i$, $u$-$u_j$, $u$-$u_k$
  and~$u$-$u_\ell$, and no other GRR is split, since the contacts
  in~$P$ are non-crossing. Therefore, it
  is~$x = x_1 + \dots + x_7 - 3$. By the definition of~$\sigma_1$, it
  is $\sigma_1[u, t_i, t_i] \leq x_1$,
  $\sigma_1[u, t_j, t_j] \leq x_2$, $\sigma_1[u, t_k, t_k]\leq x_3$
  and~$\sigma_1[u, t_\ell, t_\ell] \leq x_4$. By the definition
  of~$\sigma_M$, $\sigma_M[u, i+1, j-1] \leq x_5$,
  $\sigma_M[u, j+1, k-1] \leq x_6$
  and~$\sigma_M[u, k+1, \ell-1] \leq x_7$. Thus, the right side of
  recurrence~\eqref{noncrossing:rec:sigma4} is at most~$x$, so the
  right side is upper bounded by the left side.

  Conversely, let the right side of
  recurrence~\eqref{noncrossing:rec:sigma4} be less than~$\infty$.
  Let~$j,k$, $t_j, t_k$ be chosen such that the minimum on the right
  side is realized.  Then,
  $\rho_i \cup \rho_j \cup \rho_k \cup \rho_\ell$ is
  \ic. Let~$\sigma_1[u,t_i,t_i] = x_1$, $\sigma_1[u,t_j,t_j] = x_2$,
  $\sigma_1[u,t_k,t_k] = x_3$ and~$\sigma_1[u,t_\ell,t_\ell] = x_4$. 
  Let~$P_1$, $P_2$, $P_3$ and~$P_4$ be \gtds realizing the
  minimum in the definitions of~$\sigma_1[u,t_i,t_i]$, $\sigma_1[u,t_j,t_j]$,
  $\sigma_1[u,t_k,t_k]$ and~$\sigma_1[u,t_\ell,t_\ell]$, respectively.
  Next, let $\sigma_M[u, i+1, j-1] = x_5$,
  $\sigma_M[u, j+1, k-1] = x_6$ and~$\sigma_M[u, k+1, \ell-1] = x_7$.
  Let~$P_5$, $P_6$ and~$P_7$ be \gtds realizing the minima in
  the definitions of $\sigma_M[u, i+1, j-1]$, $\sigma_M[u, j+1, k-1]$
  and~$\sigma_M[u, k+1, \ell-1]$, respectively.  The four paths
  $\rho_i$, $\rho_j$, $\rho_k$, $\rho_\ell$ can be merged into a
  single GRR~$R$ with leftmost path~$\rho_i$ and rightmost
  path~$\rho_\ell$. Consider partition~$P$ with root component~$R$
  formed by taking the union of~$P_1$, \dots, $P_7$ and merging the
  four paths~$\rho_i$, $\rho_j$, $\rho_k$, $\rho_\ell$. No more GRRs
  can be merged, since the contacts in~$P_1$, \dots, $P_7$ are
  non-crossing.
  The GRR~$R$ is the root component of~$P$. It has leftmost and
  rightmost paths~$u$-$t_i$ and~$u$-$t_\ell$ respectively, and~$u$ has
  degree~$4$ in~$R$. Thus, by the definition
  of~$\sigma_4[u,t_i,t_\ell]$, it
  is~$\sigma_4[u,t_i,t_\ell] \leq x_1 + \dots + x_7 - 3$. Thus, the
  left side of recurrence~\eqref{noncrossing:rec:sigma4} is upper
  bounded by its right side. Therefore,
  recurrence~\eqref{noncrossing:rec:sigma4} holds.
\end{proof}

\begin{lemma}
\label{lem:noncrossing:rec:sigmaM}
  We have the following recurrence.
  \begin{compactenum}
  \item[\nextrec{noncrossing:rec:sigma5}]
    $\sigma_M[u,i,\ell] = \min_{t_j, t_k} \{ \sigma_M[u,i,j-1] +
    \sigma[u, t_j, t_k] + \sigma_M[u, k+1, \ell] \}$,

    The minimization only considers $j,k$
    for~$i \leq j \leq k \leq \ell$ and vertices~$t_j, t_k$, such
    that~$t_j$ is in~$T_{u_j}$ and~$t_k$ is in~$T_{u_k}$.
  \end{compactenum}
\end{lemma}
\begin{proof}
  First, consider a \gtd~$P$ of~$T_{u_i} \cup \dots \cup
  T_{u_\ell}$. Consider a GRR~$R$ in~$P$ containing~$u$ with leftmost
  and rightmost paths~$u$-$t_j$ and~$u$-$t_k$, respectively, for some
  vertices~$t_j$ in~$T_{u_j}$ and~$t_k$ in~$T_{u_k}$. Additionally,
  let~$R$ be chosen such that~$k-j$ is maximized. Then, by the choice
  of~$R$, no GRR in~$P$ has vertices both
  in~$T_{u_i} \cup \dots \cup T_{u_{j-1}}$ and
  in~$T_{u_{k+1}} \dots T_{u_\ell}$. Therefore, we can split
  partition~$P$ into \gtds~$P_1$
  of~$T_{u_i} \cup \dots \cup T_{u_{j-1}}$ of size~$x_1$, $P_2$ of
  $T_{u_j} \cup \dots \cup T_{u_{k}}$ of size~$x_2$ and~$P_3$
  of~$T_{u_{j+1}} \cup \dots \cup T_{u_\ell}$ size~$x_3$, such that no
  GRR of~$P$ is split. Thus,~$x = x_1 + x_2 + x_3$. By the definition
  of~$\sigma$ and~$\sigma_M$, we have $\sigma_M[u,i,j-1] \leq x_1$,
  $\sigma[u, t_j, t_k] \leq x_2$
  and~$\sigma_M[u, k+1, \ell] \leq x_3$.  Therefore, the right side of
  recurrence~\eqref{noncrossing:rec:sigma5} is at most~$x$, so the
  right side is upper bounded by the left side.

  Conversely, let the right side of
  recurrence~\eqref{noncrossing:rec:sigma5} be less than~$\infty$.
  Let~$j,k$, $t_j, t_k$ be chosen such that the minimum on the right
  side is realized. Let~$P_1$, $P_2$, $P_3$ be \gtds of
  size~$x_1$, $x_2$, $x_3$, respectively, realizing the minima in the
  definitions of $\sigma_M[u,i,j-1]$, $\sigma[u, t_j, t_k]$
  and~$\sigma_M[u, k+1, \ell]$, respectively. The union of the three
  partitions is a \gtd
  of~$T_{u_i} \cup \dots \cup T_{u_\ell}$. Thus, by the definition
  of~$\sigma_M[u,i,\ell]$, it
  is~$\sigma_M[u,i,\ell] \leq x_1 + x_2 + x_3$, so the left side of
  recurrence~\eqref{noncrossing:rec:sigma5} is upper bounded by its
  right side. Therefore, recurrence~\eqref{noncrossing:rec:sigma5}
  holds.
\end{proof}

\begin{lemma}
\label{lem:noncrossing:rec:tau}
  We have the following recurrences regarding~$\tau$.
  \begin{compactenum}
  \item[\nextrec{noncrossing:rec:sigma6}] $\tau[u,u,u] = 1 + \sigma_M[1,d]$;
  \item[\nextrec{noncrossing:rec:sigma7}]
    $\tau[u,t_i,t_j] = \sigma_M[u,1,i-1] + \sigma[u,t_i,t_j] +
    \sigma_M[u,j+1,d]$, if~$\parent{u} u + \rho_i \cup \rho_j$ is
    \ic, and~$\infty$ otherwise.

    In recurrence~\eqref{noncrossing:rec:sigma7}, vertex~$t_i \neq u$
    is in~$T_{u_i}$ and vertex~$t_j \neq u$ is in~$T_{u_j}$.
  \end{compactenum}
\end{lemma}
\begin{proof}
  First, we prove recurrence~\eqref{noncrossing:rec:sigma6}. Let~$P$
  be a \gtd of~$T_u = \parent{u}u + T_{u_1} \cup \dots \cup T_{u_d}$,
  such that the edge~$\parent{u} u$ is the root component
  of~$P$. Then, the other GRRs of~$P$ induce a partition~$P_1$
  of~$T_{u_1} \cup \dots \cup T_{u_d}$. Let~$x_1$ be the size
  of~$P_1$. Then,~$P$ has size~$x_1 + 1$. Furthermore, by the
  definition of~$\sigma_M$, $\sigma_M[u,1,d] \leq x_1$. Thus, the
  right side of recurrence~\eqref{noncrossing:rec:sigma6} is at
  most~$x_1 + 1$, so the right side is upper bounded by the left side.

  Conversely, let the right side of
  recurrence~\eqref{noncrossing:rec:sigma6} be less
  than~$\infty$. Let~$P_1$ be a \gtd
  of~$T_{u_1} \cup \dots \cup T_{u_d}$ size~$x_1$. We add
  edge~$\parent{u}u$ as a new GRR to~$P_1$ and get a partition~$P$
  of~$T_u$ of size~$x_1+1$ having~$\parent{u} u$ as its root
  component. Thus, the left side of
  recurrence~\eqref{noncrossing:rec:sigma6} is at most~$x_1 + 1$, so
  the left side is upper bounded by the right side. Therefore,
  recurrence~\eqref{noncrossing:rec:sigma6} holds.

  We now prove recurrence~\eqref{noncrossing:rec:sigma7}. Let~$P$ be a
  \gtd of~$T_u$ of size~$x$ with root component~$R$, such that~$R$ has
  $\parent{u}$-$t_i$ and~$\parent{u}$-$t_j$ as its leftmost and
  rightmost paths, respectively. Then, no GRR of~$P$ has edges both
  in~$T_{u_1} \cup \dots \cup T_{u_{i-1}}$ and
  in~$T_{u_{j+1}} \cup \dots \cup T_{u_{d}}$, since otherwise such a
  GRR would cross~$R$. Thus,~$P$ can be split into \gtds~$P_1$
  of~$T_{u_1} \cup \dots \cup T_{u_{i-1}}$ of size~$x_1$, $P_2$
  of~$\parent{u}u + T_{u_i} \cup \dots \cup T_{u_{j}}$ of size~$x_2$
  and~$P_3$ of~$T_{u_{j+1}} \cup \dots \cup T_{u_{d}}$ of size~$x_3$,
  such that~$R$ is the root component of~$P_2$ and such that it
  is~$x = x_1 + x_2 + x_3$. By the definition of~$\sigma$
  and~$\sigma_M$, we have $\sigma_M[u,1,i-1] \leq x_1$,
  $\sigma[u, t_i, t_j] \leq x_2$
  and~$\sigma_M[u, j+1, \ell] \leq x_3$. Thus, the right side of
  recurrence~\eqref{noncrossing:rec:sigma7} is at most~$x$, so the
  right side is upper bounded by the left side.

  Finally, let the right side of
  recurrence~\eqref{noncrossing:rec:sigma7} be less than~$\infty$. Let
  $P_1$ be a \gtd of of~$T_{u_1} \cup \dots \cup T_{u_{i-1}}$ of
  size~$x_1$, let $P_2$ be a \gtd
  of~$T_{u_i} \cup \dots \cup T_{u_{j}}$ of size~$x_2$ and~$P_3$ a
  \gtd of~$T_{u_{j+1}} \cup \dots \cup T_{u_{d}}$ of size~$x_3$, such
  that~$R$ is the root component of~$P_2$ having leftmost and
  rightmost paths~$u$-$t_i$ and~$u$-$t_j$,
  respectively. If~$\parent{u} u + \rho_i \cup \rho_j$ is \ic, by
  Lemma~\ref{lem:trees:combine:paths}, the subtree
  $R_2 := \parent{u} u + R$ is also a GRR. By taking the union
  of~$P_1$, $P_2$ and~$P_3$ and merging~$R$ and~$\parent{u} u$
  into~$R_2$, we get a \gtd~$P$ of~$T_u$ of
  size~$x := x_1 + x_2 + x_3$ with the root component~$R_2$, such
  that~$R_2$ has the leftmost and rightmost paths~$\parent{u} t_i$
  and~$\parent{u} t_j$, respectively. By the definition of~$\tau$, it
  is~$\tau[u,t_i,t_j] \leq x$, so the left side of
  recurrence~\eqref{noncrossing:rec:sigma7} is is upper bounded by the
  right side. Therefore, recurrence~\eqref{noncrossing:rec:sigma7}
  holds.
\end{proof}

We can now use the above recurrences to fill the tables~$\tau$,
$\sigma$, $\sigma_\Delta$ and~$\sigma_M$ in polynomial time. This
proves Theorem~\ref{thm:trees:partition-edges-non-crossing}.

\begin{reptheorem}{Theorem}{\ref{thm:trees:partition-edges-non-crossing}}
\ThmNonCrossingText
\end{reptheorem}

\begin{proof}
  For each pair~$s,t \in V$, it can be tested in time~$O(n)$ whether
  the path~$s$-$t$ is \ic~\cite{acglp-sag-12}. We store the result for
  each pair~$s,t \in V$, which allows us to query in time~$O(1)$
  whether any $s$-$t$ path is \ic. This precomputation takes $O(n^3)$
  time.

  We process the vertices~$u \in V$ bottom-up and fill the
  tables~$\tau[u, \cdot, \cdot]$, $\sigma[u, \cdot, \cdot]$,
  $\sigma_\Delta[u, \cdot, \cdot]$ and~$\sigma_M[u, \cdot,
  \cdot]$. Consider a vertex~$u \in V$ and assume all these values
  have been computed for all successors of~$u$.

  Using recurrences~\eqref{noncrossing:rec:sigma1}
  and~\eqref{noncrossing:rec:sigma11}, we can compute all values
  of~$\sigma_1[u, t_i, t_j]$ and $\sigma_M[u,i,i]$ in~$O(n^2)$
  time. We shall compute the remaining values
  $\sigma_\Delta[u, t_i, t_\ell]$, $\sigma[u, t_i, t_\ell]$
  and~$\sigma_M[u, i, \ell]$ by an induction over $\ell-i$. For a
  fixed~$m \geq 0$, assume all these values have been computed
  for~$\ell-i \leq m$. We show how to compute them for~$\ell-i = m+1$.

  First, we compute the new values~$\sigma_\Delta[u, t_i, t_\ell]$
  from the already computed ones using
  recurrences~\eqref{noncrossing:rec:sigma2}, \dots,
  \eqref{noncrossing:rec:sigma5}. This can be done in~$O(n^4)$ time by
  testing all combinations of~$t_i$, $t_j$, $t_k$, $t_\ell$. 
  Next, we compute
  $\sigma[u, t_i, t_\ell] = \min_{\Delta=1, \dots, 4}
  \sigma_\Delta[u,t_i,t_\ell]$ in~$O(n^2)$ time.
  After that, the new values~$\sigma_M[u, i, \ell]$ can be computed
  using recurrence~\eqref{noncrossing:rec:sigma5}. This can be done
  in~$O(n^4)$ time by testing all combinations of~$i$, $\ell$, $t_j$,
  $t_k$.

  In this way, we compute all values $\sigma_\Delta[u, t_i, t_\ell]$,
  $\sigma[u, t_i, t_\ell]$ and~$\sigma_M[u, i, \ell]$, for
  all~$\ell - i \leq d$, in~$O(n^5)$ time. Then, we
  compute~$\tau[u, t_i, t_j]$ using
  recurrences~\eqref{noncrossing:rec:sigma6}
  and~\eqref{noncrossing:rec:sigma7}. This can be done in~$O(n^2)$
  time by testing all combinations of~$t_i$ and~$t_j$. After that, we
  compute~$\tau[u]$. It took us~$O(n^5)$ time to compute all the
  values for the vertex~$u$.

  Let~$r$ be the root of~$T$, and let~$v$ be the only child of~$r$. By
  the above procedure, we can compute~$\tau[v]$ in~$O(n^6)$
  time. Since~$T = T_v$, $\tau[v]$ is the minimum size of a \gtd
  of~$T$.
\end{proof}

For partitions allowing edge splits, we use the results from
Section~\ref{sec:split} to reduce the problem to the scenario without
edge splits.

\begin{corollary}
  \label{cor:2}
  An optimal partition of a plane straight-line tree drawing into GRRs
  with non-crossing contacts can be computed in~$O(n^{6})$ time, if
  no edge splits are allowed, and in~$O(n^{12})$ time, if edge
  splits are allowed.
\end{corollary}

\subsubsection{Proper contacts}

For \gtds allowing only proper contacts of GRRs, we can modify the
above dynamic program. We redefine~$\sigma_M[u,i,j]$ to be the size of
a minimum \gtd of~$T_{u_i} \cup \dots \cup T_{u_j}$, in which no two
edges~$u u_i, \dots, u u_j$ are in the same GRR.  Furthermore, we
replace two recurrences as follows.

\newcounter{recnew}
\setcounter{recnew}{5}
\renewcommand{\therecnew}{\arabic{recnew}'}
\newcommand{\nextrecnew}[1]{\refstepcounter{recnew}(\arabic{recnew}')\label{#1}}
\begin{lemma}
  For \gtds with proper contacts, the following recurrences
  replace recurrences~\eqref{noncrossing:rec:sigma5}
  and~\eqref{noncrossing:rec:sigma6}.
  \begin{compactenum}
      \item[\nextrecnew{noncrossing:rec:sigma6new}]
    $\sigma_M[u,i,j] = \sum_{m = i}^j \sigma_1[u,m,m]$;
  \item[\nextrecnew{noncrossing:rec:sigma7new}]
    $\tau[u,u,u] = 1 + \min_{t_i, t_j} \{ \sigma_M[u,1,i-1] +
    \sigma[u, t_i, t_j] + \sigma_M[u, j+1, d] \}$.  

    The minimization in recurrence~\eqref{noncrossing:rec:sigma7new}
    only considers $i,j$ for~$1 \leq i \leq j \leq d$ and
    vertices~$t_i, t_j$, such that~$t_i$ is in~$T_{u_i}$ and~$t_j$ is
    in~$T_{u_j}$.
  \end{compactenum}
\end{lemma}

Recurrence~\eqref{noncrossing:rec:sigma6new} follows trivially from
the new definition of~$\sigma_M$. The proof of
recurrence~\eqref{noncrossing:rec:sigma7new} is very similar to the
proof of
Lemma~\ref{lem:noncrossing:rec:sigmaM}. Recurrences~\eqref{noncrossing:rec:sigma1},
\dots, \eqref{noncrossing:rec:sigma4}
and~\eqref{noncrossing:rec:sigma7} still hold and can be proved by
reusing the proofs of Lemma~\ref{lem:noncrossing:rec}
and~\ref{lem:noncrossing:rec:tau}. The runtime of the modified dynamic
program remains the same. This proves
Theorem~\ref{thm:trees:partition-edges}.

\begin{reptheorem}{Theorem}{\ref{thm:trees:partition-edges}}
\ThmProperText
\end{reptheorem}

Analogously as for non-crossing contacts, we use the results from
Section~\ref{sec:split} to extend the result to \gtds allowing edge
splits.

\begin{corollary}
  \label{cor:3}
  An optimal partition of a plane straight-line tree drawing into GRRs
  with proper contacts can be computed in~$O(n^{6})$ time, if no
  edge splits are allowed, and in~$O(n^{12})$ time, if edge splits
  are allowed.
\end{corollary}

Note that Corollary~\ref{cor:3} provides a better runtime than the
dynamic program in the conference version of this paper~\cite{Noellenburg2015}.

\section{Triangulations}\label{sec:triang}

In this section, we consider GRR partitions of polygonal
regions. Recall that a polygonal region is a GRR if and only if it
contains no pairs of conflicting edges. Further, recall that GRRs that
are polygonal regions need not be convex and that they do not have
holes~\cite{tk-ggrlssn-2012}. Since partitioning polygonal regions
into a minimum number of GRRs is \NP-hard~\cite{tk-ggrlssn-2012}, we
study special cases of this problem.

We consider partitioning a hole-free polygon $\mathcal P$ with a
fixed triangulation into a minimum number of GRRs by cutting it along
chords of~$\mathcal P$ contained in the triangulation. For such
decompositions we restrict the GRRs to consist of a group of
triangles of the triangulation whose union forms a simple polygon
without articulation points.
Note that allowing articulation points makes the problem \NP-hard. To
prove this, we can easily turn the plane straight-line tree
drawing~$\Gamma$ from Section~\ref{sec:trees:npc}, which is a
subdivision of a star, into a hole-free triangulated polygon with a
single articulation point corresponding to the star
center.

We reduce the problem to \minmulticut on trees and use it to give a
polynomial-time $(2-1/\opt)$-approximation, where $\opt$ is the number
of GRRs in an optimal partition.
Recall that a polygon is a GRR if and only if it has no conflict
edges~\cite{tk-ggrlssn-2012}. Let~$\triangle_{uvw}$ be the triangle
defined by three non-collinear points $u,v,w$.

\begin{lemma}
  \label{lem:add-deg-2}
  Let~$\mathcal P$ be a simple polygon, $uv$ an edge on its boundary and~$w
  \notin \mathcal P$ another point, such that $\mathcal P \cap \triangle_{uvw} =
  uv$. If~$\mathcal P$ is not a greedy region, neither is~$\mathcal P \cup
  \triangle_{uvw}$.
\end{lemma} 
\begin{proof}
  Polygon $\mathcal P' = \mathcal P \cup \triangle_{uvw}$ can become
  greedy only if~$uv$ is a conflict edge in~$\mathcal P$.
  Then, either~$u v$ is crossed by a normal ray to another edge, or a
  normal ray to~$u v$ crosses another edge. In the former case,
  either~$u w$ or~$w v$ is crossed by a normal ray to another edge, a
  contradiction to the greediness
  of~$\mathcal P \cup \triangle_{uvw}$.

  In the latter case, there exists a point~$p$ in the interior
  of~$uv$, such that~$\nray{uv}{p}$ crosses the boundary
  $\partial \mathcal P$ of~$\mathcal P$. Let~$y$ be the first
  intersection point; see Fig.~\ref{fig:add-deg-2}.  Then,
  either~$\nray{uv}{u}$ or~$\nray{uv}{v}$ must also cross
  $\partial \mathcal P$.
  \Wlg, there exists a point~$x$ on $\partial \mathcal P$, such that:
  $vx$ and~$uv$ are orthogonal, $vx \cap \mathcal P = \{v,x\}$, and
  adding edge~$vx$ to~$\mathcal P$ would create an inner face~$f$,
  such that~$u$ is not on the boundary of~$f$; see
  Fig.~\ref{fig:add-deg-2}.

  Let~$\rho$ be the $v$-$x$ path on the boundaries of
  both~$\mathcal P$ and~$f$. \Wlg, let~$uv$ point upwards, and let~$x$
  lie to the right of~$uv$. Then,~$w$ must lie to the right of the
  line through~$uv$, and there must exist a point~$q$ on~$vw$, such
  that~$\nray{vw}{q}$ intersects~$\rho$.
\end{proof}

\begin{figure}[tb]
    \hfill \subfloat[\label{fig:nonGRR-extend}]{
  \includegraphics[scale=0.8]{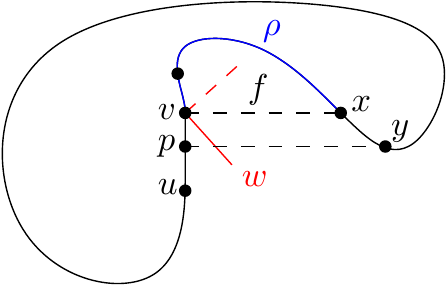}    \label{fig:add-deg-2}}
\hfill
\subfloat[]{
  \includegraphics[scale=0.8,page=2]{fig/lem-add-deg-2.pdf} \label{fig:confl-triang}}\hfill\null
\caption{ \protect\subref{fig:add-deg-2}~When adding triangles as in
  Lemma~\ref{lem:add-deg-2},~$\mathcal P$ remains non-greedy.
  \protect\subref{fig:confl-triang}~Conflicting triangles.}
\end{figure}

From now on, let~triangles $\tau_1, \dots, \tau_n$ form a
triangulation of a simple hole-free polygon~$\mathcal P$, and let~$T$ be its
corresponding dual binary tree. For simplicity we use $\tau_i$ to
refer both to a triangle in $\mathcal P$ and its dual node in $T$.

\begin{definition}[Projection of an edge]
  For three non-collinear points~$u_1,u_2,u_3$, let $\nproj{u_1}{u_2
    u_3}$ denote the set of points covered by shifting~$u_2 u_3$
  orthogonally to itself and away from~$u_1$ (blue in
  Fig.~\ref{fig:confl-triang}).
\end{definition}

\begin{definition}[Conflicting triangles]
  Let~$\tau_i = \triangle_{u_1 u_2 u_3}$ and~$\tau_j = \triangle_{v_1
    v_2 v_3}$ be two triangles such that the two edges dual
  to~$u_1u_2$ and~$v_1v_2$ are on the~$\tau_i$-$\tau_j$ path
  in~$T$. We call~$\tau_i$, $\tau_j$ \emph{conflicting}, if
  $\nproj{u_1}{u_2 u_3} \cup \nproj{u_2}{u_1 u_3}$ contains an
  interior point of~$\tau_j$.
\end{definition}

\newcommand{\lemGreedySubtreeText}{Let $T' \subset T$ be a subtree of $T$ and let $\mathcal P'$ be the
  corresponding simple polygon dual to $T'$.  Then $\mathcal P'$ is a GRR if and
  only if no two triangles $\tau, \tau'$ in $\mathcal P'$ are conflicting. }
\begin{lemma}
  \label{lem:greedy-subtree}
  \lemGreedySubtreeText
\end{lemma}
\begin{proof}
  Assume there are two conflicting triangles
  $\tau_i = \triangle_{u_1 u_2 u_3}, \tau_j= \triangle_{v_1 v_2 v_3}$
  in $T'$.  Let $\mathcal P''$ denote the polygon defined by the
  $\tau_i$-$\tau_j$ path in $T'$ and assume that the two edges dual
  to~$u_1u_2$ and~$v_1v_2$ are on the~$\tau_i$-$\tau_j$ path.  Since
  $\tau_i$ and $\tau_j$ are conflicting, there is, without loss of
  generality, a point $p$ on $u_2 u_3$ such that $\nray{u_2 u_3}{p}$
  intersects an edge of $\tau_j$. Hence, $\mathcal P''$ is not greedy.
  Moreover, $\mathcal P'$ is obtained from $\mathcal P''$ by adding
  triangles.  Thus Lemma~\ref{lem:add-deg-2} implies that
  $\mathcal P'$ cannot be greedy.

  Conversely, assume $\mathcal P'$ is not greedy.  There exists an outer edge
  $uv$ of $\mathcal P'$ and a point $x$ in the interior of $uv$ such that
  $\nray{uv}{x}$ crosses another boundary edge of $\mathcal P'$ in a point $y$.
  Let $\tau_x,\tau_y$ be the triangles with $x \in \tau_x$ and $y \in
  \tau_y$. Then $\tau_x$ and $\tau_y$ are conflicting.
\end{proof}
By Lemma~\ref{lem:greedy-subtree}, the decompositions
of $\mathcal P$ in $k$ GRRs correspond bijectively to the multicuts $E'$ of
$T$ with $|E'| = k-1$ where the terminal pairs are the pairs of
conflicting triangles.

We now use the 2-approximation for \minmulticut on
trees~\cite{gvy-pdaaifmt-1997} to give a $(2-1/\opt)$-approximation
for the minimum GRR decomposition of $\mathcal P$.  Let $E'$ be a 2-approximation of
\minmulticut in $T$ with respect to the pairs of conflicting
triangles.  By the above observation the minimum multicut for $T$ has
size $\opt-1$, hence $|E'| \le 2\opt-2$, which in turn yields a
decomposition into $2\opt-1$ regions.  Thus the approximation
guarantee is $2-1/\opt$.  
We summarize this in Theorem~\ref{thm:triang:decomp}.

\begin{theorem}
  There is a polynomial-time $(2-1/\opt)$-approximation for minimum
  GRR decomposition of triangulated simple polygons.
\label{thm:triang:decomp}
\end{theorem}

\section{Conclusions} 

Motivated by a geographic routing protocol for dense wireless sensor
networks proposed by Tan and Kermarrec~\cite{tk-ggrlssn-2012}, we
further studied the problem of finding minimum GRR decompositions of
polygons. We considered the special case of decomposing plane
straight-line drawings of graphs, which correspond to infinitely thin
polygons.
For this case, we could apply insights gained from the study of
self-approaching and increasing-chord drawings by the graph drawing
community.

We extended the result of Tan and Kermarrec~\cite{tk-ggrlssn-2012} for
polygonal regions with holes by showing that partitioning a plane
graph drawing into a minimum number of increasing-chord components is
\NP-hard.
We then considered plane drawings of trees and showed how to model the
decomposition problem using \minmulticut, which provided a
polynomial-time 2-approximation.
We solved the partitioning problem
for trees optimally in polynomial time using dynamic programming.
Finally, using insights gained from the decomposition of graph
drawings, we gave a polynomial-time 2-approximation for decomposing
triangulated polygons along their chords.

\subsubsection*{Open questions}

For the \NP-hard problem of decomposing plane drawings of graphs
into the minimum number of GRRs, it is interesting to find
approximation algorithms.

For decomposing polygons, many problems remain open. For example, one
could investigate whether minimum decomposition is \NP-hard for simple
polygons for different types of allowed partition types.
Is finding the optimum solution hard for partitioning triangulations
as in Section~\ref{sec:triang}? Is the minimum GRR decomposition
problem hard if we allow cutting the polygon at any diagonal? Is it
hard if arbitrary polygonal cuts are allowed, i.e., the partition can
use Steiner points?
Finally, are there approximations for partitioning polygons with and
without holes into GRRs?

\subsubsection*{Acknowledgements} The second author thanks Jie Gao for
pointing him to the topic of GRR decompositions.

 {\small

\begin{thebibliography}{10}

\bibitem{acglp-sag-12}
S.~Alamdari, T.~M. Chan, E.~Grant, A.~Lubiw, and V.~Pathak.
\newblock Self-approaching graphs.
\newblock In W.~Didimo and M.~Patrignani, editors, {\em GD 2012}, volume 7704
  of {\em LNCS}, pages 260--271. Springer, 2013.

\bibitem{bmsu-rgdahwn-01}
P.~Bose, P.~Morin, I.~Stojmenović, and J.~Urrutia.
\newblock Routing with guaranteed delivery in ad hoc wireless networks.
\newblock {\em Wireless Networks}, 7(6):609--616, 2001.

\bibitem{gfr-mugd-2003}
G.~Calinescu, C.~G. Fernandes, and B.~Reed.
\newblock Multicuts in unweighted graphs and digraphs with bounded degree and
  bounded tree-width.
\newblock {\em J. Algorithms}, 48(2):333--359, 2003.

\bibitem{cd-ocd-85}
B.~Chazelle and D.~Dobkin.
\newblock Optimal convex decompositions.
\newblock In {\em Computational Geometry}, pages 63--133, 1985.

\bibitem{cv-svht-07}
D.~Chen and P.~K. Varshney.
\newblock A survey of void handling techniques for geographic routing in
  wireless networks.
\newblock {\em Commun. Surveys Tuts.}, 9(1):50--67, 2007.

\bibitem{clr-gamimrt-03}
M.-C. Costa, L.~L{\'{e}}tocart, and F.~Roupin.
\newblock A greedy algorithm for multicut and integral multiflow in rooted
  trees.
\newblock {\em Oper. Res. Lett.}, 31(1):21--27, 2003.

\bibitem{clr-mmmims-05}
M.-C. Costa, L.~Létocart, and F.~Roupin.
\newblock Minimal multicut and maximal integer multiflow: A survey.
\newblock {\em European Journal of Operational Research}, 162(1):55 -- 69,
  2005.

\bibitem{custic_geometric_2015}
A.~{\'C}usti{\'c}, B.~Klinz, and G.~J. Woeginger.
\newblock Geometric versions of the three-dimensional assignment problem under
  general norms.
\newblock {\em Discrete Optimization}, 18:38--55, 2015.

\bibitem{dfg-icgps-15}
H.~R. {Dehkordi}, F.~{Frati}, and J.~{Gudmundsson}.
\newblock Increasing-chord graphs on point sets.
\newblock {\em J. Graph Algorithms Appl.}, 19(2):761--778, 2015.

\bibitem{Glider2005}
Q.~Fang, J.~Gao, L.~Guibas, V.~de~Silva, and L.~Zhang.
\newblock Glider: gradient landmark-based distributed routing for sensor
  networks.
\newblock In {\em INFOCOM 2005}, pages 339--350. IEEE, 2005.

\bibitem{gvy-pdaaifmt-1997}
N.~Garg, V.~Vazirani, and M.~Yannakakis.
\newblock Primal-dual approximation algorithms for integral flow and multicut
  in trees.
\newblock {\em Algorithmica}, 18(1):3--20, 1997.

\bibitem{ikl-sac-99}
C.~Icking, R.~Klein, and E.~Langetepe.
\newblock Self-approaching curves.
\newblock {\em Math. Proc. Camb. Phil. Soc.}, 125:441--453, 1999.

\bibitem{k-dpsc-85}
J.~M. Keil.
\newblock Decomposing a polygon into simpler components.
\newblock {\em SIAM Journal on Computing}, 14(4):799--817, 1985.

\bibitem{kr-pcr-92}
D.~E. Knuth and A.~Raghunathan.
\newblock The problem of compatible representatives.
\newblock {\em SIAM J. Discrete Math.}, 5(3):422--427, 1992.

\bibitem{kranakis_compass_1999}
E.~Kranakis, H.~Singh, and J.~Urrutia.
\newblock Compass routing on geometric networks.
\newblock In {\em Canadian Conference on Computational Geometry (CCCG'99)},
  pages 51--54, 1999.

\bibitem{l-pftu-82}
D.~Lichtenstein.
\newblock Planar formulae and their uses.
\newblock {\em SIAM Journal on Computing}, 11(2):329--343, 1982.

\bibitem{mwh-spbr-01}
M.~Mauve, J.~Widmer, and H.~Hartenstein.
\newblock A survey on position-based routing in mobile ad hoc networks.
\newblock {\em Network, IEEE}, 15(6):30--39, 2001.

\bibitem{Noellenburg2015}
M.~N{\"o}llenburg, R.~Prutkin, and I.~Rutter.
\newblock Partitioning graph drawings and triangulated simple polygons into
  greedily routable regions.
\newblock In K.~Elbassioni and K.~Makino, editors, {\em Algorithms and
  Computation (ISAAC'15)}, volume 9472, pages 637--649. Springer, 2015.

\bibitem{npr-sid3pg-16}
M.~N{\"o}llenburg, R.~Prutkin, and I.~Rutter.
\newblock On self-approaching and increasing-chord drawings of 3-connected
  planar graphs.
\newblock {\em J. Comput. Geom.}, 7(1):47--69, 2016.

\bibitem{Papadimitriou2005}
C.~H. Papadimitriou and D.~Ratajczak.
\newblock On a conjecture related to geometric routing.
\newblock {\em Theoret. Comput. Sci.}, 344(1):3--14, 2005.

\bibitem{tbk-cpsn-09}
G.~Tan, M.~Bertier, and A.-M. Kermarrec.
\newblock Convex partition of sensor networks and its use in virtual coordinate
  geographic routing.
\newblock In {\em INFOCOM 2009}, pages 1746--1754. IEEE, 2009.

\bibitem{tk-ggrlssn-2012}
G.~Tan and A.-M. Kermarrec.
\newblock Greedy geographic routing in large-scale sensor networks: A minimum
  network decomposition approach.
\newblock {\em IEEE/ACM Transactions on Networking}, 20(3):864--877, 2012.

\bibitem{zsg-ssasn-07}
X.~Zhu, R.~Sarkar, and J.~Gao.
\newblock Shape segmentation and applications in sensor networks.
\newblock In {\em INFOCOM 2007}, pages 1838--1846. IEEE, 2007.

\end{thebibliography}

}

\end{document}